\documentclass[10pt,a4paper]{article}
\usepackage{amssymb}
\usepackage{amsmath}
\usepackage{amsthm}
\usepackage{latexsym}
\usepackage[dvips]{epsfig}
\usepackage{mathrsfs}
\usepackage{eufrak}
\usepackage{bm}

\theoremstyle{plain}
\newtheorem{proposition}{Proposition}
\newtheorem{lemma}{Lemma}
\newtheorem{theorem}{Theorem}

\newtheorem{observation}{Observation}

\newtheorem{corollary}{Corollary}
\newtheorem*{main}{Theorem}
\newtheorem{definition}{Definition}

\font\SYM=msbm10
\newcommand{\Real}{\mbox{\SYM R}}
\newcommand{\Complex}{\mbox{\SYM C}}

\newcommand{\Sphere}{\mbox{\SYM S}}

\font\tenscr=rsfs10 scaled1100
\font\sevenscr=rsfs7 % scaled \magstep1
\font\fivescr=rsfs5 % scaled \magstep1
\skewchar\tenscr='177
\skewchar\sevenscr='177
\skewchar\fivescr='177
\newfam\scrfam
\textfont\scrfam=\tenscr
\scriptfont\scrfam=\sevenscr
\scriptscriptfont\scrfam=\fivescr

\def\pb{p_\bullet}

\newcommand{\TT}[3]{T_{#1 \phantom{#2} #3}^{\phantom{#1} #2}}
\newcommand{\updn}[3]{#1^{#2}_{\phantom{#2}#3}}
\newcommand{\dnup}[3]{#1_{#2}^{\phantom{#2}#3}}

\begin{document}

% Ludovica
%\bibliographystyle{/Users/Juan/Documents/tex/reporthack}
%QM
%\bibliographystyle{/home/network/jav/tex/reporthack}

\title{\textbf{Asymptotic simplicity and static data}}

\author{{\Large Juan Antonio Valiente Kroon} \thanks{E-mail address:
 {\tt j.a.valiente-kroon@qmul.ac.uk}} \\
School of Mathematical Sciences, \\ Queen Mary, University of London,\\
Mile End Road, London E1 4NS, \\ United Kingdom.}

\maketitle

\begin{abstract}
The present article considers time symmetric initial data sets for the
vacuum Einstein field equations which in a neighbourhood of infinity
have the same massless part as that of some static initial data set. It is shown
that the solutions to the regular finite initial value problem at
spatial infinity for this class of initial data sets extend smoothly
through the critical sets where null infinity touches spatial infinity
if and only if the initial data sets coincide with static data in a
neighbourhood of infinity. This result highlights  the
special role played by static data among the class of initial data
sets for the Einstein field equations whose development gives rise to
a spacetime with a smooth conformal compactification at null infinity.
\end{abstract}


%PACS: 

\section{Introduction}
The analysis of the structure of spatial infinity using the
\emph{Einstein conformal field equations} and the construction known
as the \emph{cylinder at spatial infinity} ---see
\cite{Fri98a,Fri04,Val04a,Val04d,Val04e,Val05a,Val10}---suggests that
static initial data sets play a special role among the class of initial
data sets for the vacuum Einstein field equations whose development
has a smooth conformal compactification at null infinity. This issue
is of fundamental importance in the construction of so-called
\emph{asymptotically simple spacetimes}. In connection with this
expectation, the following theorem has been proved in \cite{Val10}:

\begin{theorem}
\label{Theorem:ConformallyFlat}
Consider a time symmetric initial data set for the Einstein vacuum
field equations which is conformally flat near infinity. The solutions
to the regular finite initial value problem at spatial infinity is
smooth through the critical sets where null infinity touches spatial
infinity if and only if the data is exactly Schwarzschildean in a
neighbourhood of infinity.
\end{theorem}

The context of this theorem is better appreciated if one recalls that
the Schwarzschild spacetime is the only static spacetime with
conformally flat slices ---see \cite{Fri04}. The use of conformally
flat data sets in the analysis of \cite{Val10} is a convenient technical
assumption. Otherwise, the full
complexity of time symmetric initial data sets would make it
impossible to identify useful structures. The analysis in
\cite{Val10} builds upon the original analysis in \cite{Fri98a} and the
computer algebra calculations of \cite{Val04a}, to make generic
assertions about a certain type of asymptotic expansions for the
conformal Einstein field equations made possible by the framework of
the cylinder at spatial infinity.

\medskip
The purpose of the present article is to consider an extension of
Theorem \ref{Theorem:ConformallyFlat} to the case of more general
---non-conformally flat--- time symmetric initial data sets. Again,
explicit computer algebra calculations carried out in \cite{Val04e}
suggest the type of result that one can expect to prove. A
fundamental difficulty in the way of constructing a generalisation to
Theorem \ref{Theorem:ConformallyFlat} is to obtain a parametrisation of time
symmetric initial data sets for which it is simple to decide whether
the data under consideration is static or not. 

\medskip
The properties of time symmetric initial data sets for the vacuum
Einstein field equations, $(\tilde{\mathcal{S}},\tilde{h}_{ij})$, to be analysed in
the present article are best discussed using a point compactification
of the initial hypersurface $\mathcal{S}$ and an associated
conformally rescaled metric $h_{ij}$ ---the \emph{conformal
  metric}. The 3-manifold $\mathcal{S}$ contains singled out points $\{i_1,\,i_2,\dots\}$
representing the points at infinity of the 3-manifold
$\tilde{\mathcal{S}}$. Our analysis will be local to one of these
points, so it will be assumed, without loss of generality, that there
is only one of them. This point will be deonted by $i$. In \cite{Fri88} it has
been shown that static initial data sets satisfy a certain
\emph{regularity condition} involving the Cotton tensor and its higher
order derivatives ---see equation \eqref{RegularityCondition} in the
main text. This property is,
however, not enough to fully assert whether an initial data set is
static ---initial data sets like those of Misner \cite{Mis63} and
Brill-Lindquist \cite{BriLin63} 
satisfy the regularity condition as
they are conformally flat, but clearly they do not, in general, give
rise to static spacetimes. Further conditions need to be imposed on the initial
data to obtain a static development. This gap between initial data sets satisfying the
regularity condition and initial data sets which are exactly static at
spatial infinity is still to be understood\footnote{Recently, there
has been some progress in this direction ---H. Friedrich, parallel
session talk in the GR19 Conference, Mexico.}. 

\medskip
The present article gets around the difficulty exposed in the previous
paragraph by considering a restricted class of initial data sets for
which it is simple to decide whether they are actually static or
not. This class of initial data sets is constructed by looking at
solutions to the equation giving rise to the conformal factor
$\vartheta$ relating the 3-metric $\tilde{h}_{ij}$ and $h_{ij}$, the
\emph{Yamabe equation} ---see equation \eqref{YamabeEquation} in the
main text. In a suitably small neighbourhood of infinity, the
solutions to the Yamabe equation can be split into its \emph{massless}
and \emph{massive} parts. The massless part contains the information
of the local geometry in a neighbourhood of $i$, whereas the massive
part contains information of global nature---in particular the
mass. The class of initial data sets to be used in the present article
takes the solution of the Yamabe equation for static data and adds to
it a further massive term that does not contribute to the mass so to
obtain a new solution to the Yamabe equation ---this can be done
because of the linearity of the setting. This new solution to the
Yamabe equation implies, in turn, a new solution to the constraint
equations in a neighbourhood of infinity with the same conformal
metric $\tilde{h}_{ij}$ as a static initial data set. We say that these
solutions to the constraint equations have a \emph{static
massless part}.  It can be verified that this class of initial data
sets satisfies the regularity condition of \cite{Fri88} ---cfr. also
equation \eqref{RegularityCondition} of the main text. This
observation is of relevance in the present article as it has been
shown in \cite{Fri98a} that this condition is necessary for solutions
to the regular finite initial value problem of the conformal Einstein
field equations to extend smoothly through the critical sets where
null infinity touches spatial infinity.

\medskip
For the class of time symmetric initial data sets for the Einstein
vacuum field equations discussed in the previous paragraph
one can prove the following generalisation of Theorem
\ref{Theorem:ConformallyFlat}:

\begin{main}
Given an initial data set with static massless part, the solution to the
regular finite initial value problem at spatial infinity is smooth
through the critical sets if and only if the data is exactly static in
a neighbourhood of infinity.
\end{main}

In other words, the smoothness of the development through the
critical sets forces the extra massive part added to the (background)
static initial data to vanish. A more precise version of this result
will be given in the main text.

\medskip
As in the case of the assumption of conformal flatness made in
\cite{Val10}, the use of initial data sets with a static massless part is a useful
technical assumption which allows to identify relevant structures in
the conformal field equations. It is clear that not all time symmetric
initial data sets admitting an analytic conformal compactification at
infinity have a static massless part. A general version of the main theorem of this
article can only be obtained once one knows what extra conditions have
to be imposed on a generic time symmetric initial data set to have a
static massless part ---cfr. similar remarks in the previous paragraphs. This
task requires learning how to exploit to the maximum extent the
conformal gauge freedom implicit in the conformal Ansatz. As the
conformal metric, $h_{ij}$, encodes all the freely specifiable
information of a time symmetric initial data set, the extra requirements
will have to be in the form of conformally invariant conditions on the
conformal class. 

\medskip
The proof of the main theorem builds upon the analysis of the
conformally flat case discussed in \cite{Val10}. This analysis relied
heavily on the use of computer algebra methods to transform the
transport equations implied by the conformal field equations at the
cylinder at spatial infinity into a system of reduced ordinary
differential equations for which explicit solutions can be computed
for \emph{any} order of the expansion. The approach in the present
article consists in conveniently grouping the various terms appearing
in the transport equations in \emph{Schwarzschildean terms} and
\emph{deviations-from-Schwarzschild terms}. The former are formally
identical to terms appearing in \cite{Val10} and thus, assertions
about their smoothness can be readily given. It turns out that most of
the terms that one needs to consider are Schwarzschildean terms. It is
only in the last step of the argument that deviation terms arise. As
it will be seen, their presence indicates that the extra massive part
that has been added to the static data has to vanish up to a certain
order ---thus, putting into action an inductive argument from which
the main theorem is obtained. Remarkably, essentially all the computer
algebra required for this argument has already been performed in
\cite{Val10}.

\subsection*{Outline of the article}
Section \ref{Section:Data} discusses some general
properties of time symmetric solutions to the Einstein constraint
equations in the conformal setting. It also introduces the class of
time symmetric initial data sets that will be used in our subsequent
analysis ---initial data sets with a static massless part. Section
\ref{Section:CylinderAtSpatialInfinity} gives a concise summary of the
framework of the cylinder at spatial infinity and of the so-called F-gauge. Its
main purpose is to introduce the notation to be used in the rest of
the article. It also provides an overview of the key properties of the
transport equations implied by the conformal Einstein field equations
at spatial infinity. Section \ref{Section:StaticCylinder} briefly
discusses the key result of the construction of the cylinder at
spatial infinity for static spacetimes ---namely, that the structures
are as smooth as they can be. Section
\ref{Section:DataStaticUpToAnOrder} discusses key properties of
initial data sets which are static up to a certain order ---the
results will be used extensively in the sequel. Section
\ref{Section:CoreAnalysis} contains the core of our analysis: a
discussion of the properties of solutions to the transport equations
at the cylinder at spatial infinity for data which is static up to a
certain order. The results presented in this section take the form of
an induction argument which leads, ultimately, to our main theorem in
Section \ref{Section:Conclusions}.

\subsection*{Notation and conventions}

The present article is concerned
with the asymptotic properties of spacetimes
$(\tilde{\mathcal{M}},\tilde{g}_{\mu\nu})$ solving the Einstein vacuum
field equations 
\begin{equation}
\tilde{R}_{\mu\nu}=0.
\label{EinsteinFieldEquations}
\end{equation}
 The metric $\tilde{g}_{\mu\nu}$ will be assumed to have signature
$(+,-,-,-)$ and $\mu, \,\nu,\dots$ are spacetime indices taking the
values $0,\dots,3$. The spacetime
$(\tilde{\mathcal{M}},\tilde{g}_{\mu\nu})$ will be thought of as the
development of a time symmetric initial data set
$(\tilde{\mathcal{S}},\tilde{h}_{ij})$ where $\tilde{\mathcal{S}}$ is
an asymptotically Euclidean hypersurface. The metric $\tilde{h}_{ij}$
will be taken to have signature $(-,-,-)$. The indices $i,\,j,\ldots$
will be spatial ones taking the values $1,\,2,\,3$.  The spinorial
conventions of \cite{PenRin84,PenRin86} will be adopted. The present
article draws heavily from the analysis in \cite{Fri98a,Fri04,Val10}
so we have followed the notation and conventions of these references
as closely as possible.

\section{A class of time symmetric data}
\label{Section:Data}

For time symmetric initial data sets $(\tilde{\mathcal{S}},\tilde{h}_{ij})$ the Einstein
vacuum field equations \eqref{EinsteinFieldEquations} imply the constraint equation
\begin{equation}
\label{HamiltonianConstraint}
\tilde{r}=0, \quad \mbox{on } \tilde{\mathcal{S}},
\end{equation}
where $\tilde{r}$ denotes the Ricci scalar of the metric
$\tilde{h}_{ij}$. 

\medskip
Our analysis will be local to a neighbourhood of infinity. Hence,
without loss of generality, only one asymptotically flat end will be
assumed. The asymptotic flatness of the time symmetric initial data
$(\tilde{\mathcal{S}},\tilde{h}_{ij})$ will be expressed in terms of
conditions on a conformally rescaled manifold. For this, it will be
assumed that there is a 3-dimensional, orientable, smooth compact
manifold $\mathcal{S}$, a metric $h_{ij}$, a point $i\in \mathcal{S}$,
a diffeomorphism $\Phi: \mathcal{S}\setminus\{i\}\rightarrow
\tilde{\mathcal{S}}$ and a function $\Omega \in C^2(\mathcal{S})\cap
C^\infty(\mathcal{S}\setminus\{i\})$ with the properties
\begin{subequations}
\begin{eqnarray}
&& \Omega(i)=0, \quad D_j\Omega(i)=0, \quad
D_jD_k\Omega(i)=-2h_{jk}(i), \label{Asymptotic:1}\\
&& \Omega>0 \mbox{ on } \mathcal{S}\setminus\{i\}, \label{Asymptotic:2}\\
&& h_{ij} = \Omega^2 \Phi_* \tilde{h}_{ij}, \label{Asymptotic:3}
\end{eqnarray}
\end{subequations}
where $D_j$ denotes the Levi-Civita covariant derivative of the
3-metric $h_{ij}$.  For the sake of simplicity the last condition will
be written as $h_{ij} = \Omega^2 \tilde{h}_{ij}$ so that
$\mathcal{S}\setminus\{ i\}$ is identified with
$\tilde{\mathcal{S}}$. If assumptions
\eqref{Asymptotic:1}-\eqref{Asymptotic:3} are satisfied, the pair
$(\tilde{S},\tilde{h}_{ij})$ will be said to be \emph{asymptotically
Euclidean} and regular. Suitable punctured neighbourhoods of the point
$i$ are mapped to the asymptotic end of $\tilde{\mathcal{S}}$.

\subsection{Asymptotically Euclidean and regular data}

The Hamiltonian constraint, equation \eqref{HamiltonianConstraint},
together with the boundary conditions
\eqref{Asymptotic:1}-\eqref{Asymptotic:3} imply on $\mathcal{B}_a(i)$
the \emph{Yamabe Equation}
\begin{equation}
\label{YamabeEquation}
\left( \Delta_h -\frac{1}{8} r\right) \vartheta =-4\pi \delta(i), \quad
  \vartheta \equiv \Omega^{-2},
\end{equation}
where $\delta(i)$ denotes the Dirac delta distribution with support on
$i$ while $\Delta$ and $r$ correspond, respectively, to the Laplacian
and the Ricci scalar of the conformal metric $h_{ij}$.  For later use we define
\begin{equation}
\label{omega}
\omega= \frac{2\Omega}{\sqrt{|D_k \Omega D^k\Omega|}}.
\end{equation}

\medskip
It is well known ---see e.g. \cite{Fri98a,Fri04}--- that if $a$ is
suitably small, then the solutions to the Yamabe equation
\eqref{YamabeEquation} on $\mathcal{B}_a(i)$ are of the form
\[
\vartheta = \frac{U}{|x|} +W, \qquad |x|=(
(x^1)^2+(x^2)^2+(x^3)^2)^{1/2},
\]
for some asymptotically Cartesian coordinates $x^i$. The terms
$U/|x|$ and $W$ will be known, respectively, as the \emph{massless} and
\emph{massive} parts of $\vartheta$. The function $U$, the \emph{Green
  function}, satisfies the equation 
\[
\left( \Delta_h -\frac{1}{8} r\right) \left(\frac{U}{|x|} \right) =-4\pi \delta(i),
\]
and describes the local geometry in $\mathcal{B}_a(i)$. The function
$W$ satisfies the equation 
\[
\left( \Delta_h -\frac{1}{8} r\right)W=0,
\]
annd contains information of global nature. In particular, $W(i)=m/2$,
where $m$ is the ADM mass of the initial data set. 

\medskip
A rescaling 
\[
h_{ij} \mapsto h'_{ij} =\theta^4 h_{ij}, \quad \Omega \mapsto \Omega'=\theta^2 \Omega,
\]
with a smooth positive factor $\theta$ satisfying $\theta(i)=1$ leaves
the physical metric $\tilde{h}_{ij}=\Omega^{-2}h_{ij}$
unchanged. However, it implies the transitions
\[
\vartheta\mapsto \vartheta'=\theta^{-1}\vartheta, \quad U\mapsto U'=\frac{|x'|}{|x|} \theta^{-1}U, \quad W\mapsto W'=\theta^{-1}W,
\]
where $|x'|$ is given in terms of the $h'$-normal coordinates. There
are several possibilities to fix this conformal gauge freedom. For the
purpose of the present analysis it turns out that a good choice is
that of the so-called \emph{conformal normal (cn)-gauge} introduced in
\cite{Fri98a}. In the following definition let $l_{ij}$ denote the Schouten tensor
of the metric $h_{ij}$.

\begin{definition}
The metric $h_{ij}$ is said to be in the cn-gauge if given a solution
$(x(\lambda),b(\lambda))$ to the 3-dimensional conformal geodesic
equations
\begin{eqnarray*}
&& \dot{x}^\nu D_{\nu} \dot{x}^\mu = -2  (b_\nu \dot{x}^\nu)\,
\dot{x}^\mu + (h_{\lambda\rho}\dot{x}^\lambda\dot{x}^\rho)\, h^{\mu\nu} b_\nu, \\
&& \dot{x}^\nu D_\nu b_\mu =( b_\nu \dot{x}^\nu)\, b_\mu -\tfrac{1}{2}
(h^{\lambda\rho}b_\lambda b_\rho)\, h_{\mu\nu}\dot{x}^\nu + l_{\lambda\mu}\dot{x}^\lambda,
\end{eqnarray*}
with initial conditions
\[
x(0)=i, \quad h(\dot{x},\dot{x})(i)=-1, \quad b(0)=0,
\]
one has that
\[
\langle b, \dot{x}\rangle =0.
\]   
\end{definition} 

\noindent
\textbf{Remark.} Let $h_{ij}$ be analytic in a neighbourhood of
$i$. Assuming that $a$ is sufficiently small, there exists on
$\mathcal{B}_a(i)$ a unique analytic rescaling $h_{ij}\mapsto
h'_{ij}=\theta^4h_{ij}$ for which $h'_{ij}$ is in the cn-gauge. The
metric and connection remain unchanged at $i$. As it will be discussed
in the sequel, the practical advantage of the cn-gauge is that it
renders simpler multipolar expansions for various quantities of
interest.

\subsection{Asymptotically static data}
Let $(\mathcal{S}, \mathring{h}_{ij})$ denote a static initial data
set given in the cn-gauge, and let 
\[
\mathring{\vartheta}\equiv \frac{\mathring{U}}{|x|} + \mathring{W},
\]
denote the corresponding solution to the Yamabe equation \eqref{YamabeEquation} in a suitably small neighbourhood $\mathcal{B}_a(i)$. The
static initial data set can be specified entirely in terms of
its multipole moments ---see e.g. \cite{Fri07}.  This important fact
will, however, not be exploited here. As a consequence of the analysis in
\cite{Fri04,Fri07} one has the following:

\begin{proposition}
Let $(\tilde{\mathcal{S}}, \tilde{h}_{ij})$ be a static initial data set
and let $\mathcal{B}_a(i)$ be a suitably small neighbourhood of
$i$. If the conformal metric $\mathring{h}_{ij}$ satisfies the cn-gauge $\mathcal{B}_a(i)$, then there
exist normal coordinates $x^i$ such that $\mathring{h}_{ij}$,
$\mathring{U}$ and $\mathring{W}$ are analytic in the neighbourhood.
\end{proposition}

One also has that ---cfr. \cite{Bei91b,Fri98a}---:

\begin{proposition}
\label{Proposition:RegularityCondition}
The Cotton-Bach tensor $\mathring{b}_{ij}$ of the conformal metric
$\mathring{h}_{ij}$ of a static initial data set $(\tilde{\mathcal{S}}, \tilde{h}_{ij})$
satisfies the regularity condition
\begin{equation}
\label{RegularityCondition}
\mathcal{C}(D_{k_q}\cdots D_{k_1} \mathring{b}_{kl})(i)=0, \quad q=0,\,1,\dots.
\end{equation}
\end{proposition}

\noindent
\textbf{Remark 1.} Given a sequence of multipoles for a static
solution, one can always assume without loss of generality, that the
dipolar terms vanish ---this amounts to the gauge choice of working in
the \emph{centre of mass}. It can be readily verified that this
assumption has the consequence that static data in the cn-gauge
satisfies
\[
\mathring{W} = \frac{m}{2} + \mathcal{O}(|x|^2). 
\]
The latter form of the function $\mathring{W}$ will be assumed in the sequel.

\medskip
\noindent
\textbf{Remark 2.} If $\mathring{U}=1$ and $\mathring{W}=m/2$ in
$\mathcal{B}_{a}(i)$, then the static initial data set corresponds to
initial data for the Schwarzschild spacetime.

% \medskip
% A property of static initial data which will be crucial for our
% analysis is the following ---see \cite{Fri88,Fri98a}:

% \begin{proposition}
% Let $(\mathcal{S},\mathring{h}_{ij})$ be a static initial data
% set. One has that
% \[
% \mathcal{C}(D_{i_q}\cdots D_{i_1} \mathring{b}_{ij})(i)=0, \quad q=0,\dots,p,
% \] 
% where $\mathcal{C}$ denotes the operation of taking the symmetric
% trace-free part.
% \end{proposition}

\bigskip
The following observation will be crucial in our subsequent
analysis. Let $\breve{W}$ satisfy
\[
\left( \Delta_h -\frac{1}{8} r\right)\breve{W}=0, \quad \breve{W}(i)=0.
\]
Clearly, due to linearity one has that 
\begin{equation}
\label{theta:perturbed}
\vartheta = \frac{\mathring{U}}{|x|} + \mathring{W}+\breve{W}
\end{equation}
is also a solution to the Yamabe equation \eqref{YamabeEquation} with
the same boundary conditions in $\mathcal{B}_a(i)$. Due to the
analyticity of $\mathring{h}_{ij}$, the function $\breve{W}$ will also
be analytic. Time symmetric initial data sets,
$(\tilde{\mathcal{S}},\tilde{h}_{ij})$, constructed out of the
  conformal metric $\mathring{h}_{ij}$ and the conformal factor
  $\vartheta$ given by equation \eqref{theta:perturbed} will be said
  to have a \emph{static massless part}. Moreover, following the ideas in \cite{Val10}, we make the
following \emph{ad hoc} definition:

\begin{definition}
\label{StaticUpToCertainOrder}
A time symmetric initial data set in the cn-gauge will be said to be static up to
order $\pb$ if and only if there exists coordinates $x^i$ in a suitable
neighbourhood $\mathcal{B}_a(i)$ of infinity such that the solution
$\vartheta$ to the Yamabe equation \eqref{YamabeEquation}
is of the form
\begin{equation}
\label{breveW}
\vartheta = \frac{\mathring{U}}{|x|} + \mathring{W} + \breve{W}, \quad
\breve{W}= \sum_{p=\pb+1}^\infty \breve{w}_{i_1\cdots i_{p}} x^{i_1}\cdots x^{i_{p}},
\end{equation}
where $\mathring{U}/|x|$ and $\mathring{W}$ are, respectively, the massless and
massive parts of a static initial data set. 
\end{definition}

\medskip
\noindent
\textbf{Remark.} It can be shown that in the cn-gauge the polynomials
\[
\breve{w}_{i_1\cdots i_{\pb+1}} x^{i_1}\cdots x^{i_{\pb+1}}, \quad \breve{w}_{i_1\cdots i_{\pb+2}} x^{i_1}\cdots x^{i_{\pb+2}}, \quad \breve{w}_{i_1\cdots i_{\pb+3}} x^{i_1}\cdots x^{i_{\pb+3}}
\]
are harmonic with respect to the flat Laplacian. That is, the constant tensors
\[
\breve{w}_{i_1\cdots i_{\pb+1}}, \quad \breve{w}_{i_1\cdots i_{\pb+2}}, \quad \breve{w}_{i_1\cdots i_{\pb+3}}
\]
are trace-free with respect to the flat metric $\delta_{ij}$.

\subsection{Properties of data which is static up to a certain order}
 Consistent with Definition \ref{StaticUpToCertainOrder}, we write
\[
\Omega= \mathring{\Omega} + \breve{\Omega}, \quad \omega=\mathring{\omega}+\breve{\omega}, 
\]
where $\mathring{\Omega}$ and $\mathring{\omega}$ are, respectively,
the \emph{static parts} of $\Omega$ and $\omega$ (obtained by setting
$\breve{W}=0$). Their \emph{non-static parts} $\breve{\Omega}$ and
$\breve{\omega}$ are then obtained via
\[
\breve{\Omega}=\Omega-\mathring{\Omega}, \quad \breve{\omega}=\omega-\mathring{\omega}.
\]
The following result is obtained by direct computation observing
expression \eqref{theta:perturbed}:

\begin{lemma}
Let $(\tilde{\mathcal{S}},\tilde{h}_{ij})$ be an initial data set which is static up
to order $\pb\geq2$. Then 
\[
\breve{\Omega}= \mathcal{O}(|x|^{\pb+4}), \quad \breve{\omega}= \mathcal{O}(|x|^{\pb+3}).
\]
\end{lemma}

\medskip
Let $b_{ij}$ denote  the Cotton-Bach tensor of the conformal metric
$h_{ij}$ of an initial data set which is static up to a certain order
in the sense of Definition \ref{StaticUpToCertainOrder}. Crucially, one has that:

\begin{lemma}
\label{Lemma:Cotton}
Let $(\tilde{\mathcal{S}},\tilde{h}_{ij})$ be an initial data set which is static up to order $\pb\geq 2$. Then 
\[
\mathcal{C}(D_{i_q}\cdots D_{i_1} b_{ij})(i)=0, \quad q=0,\dots,\pb,
\] 
where $\mathcal{C}$ denotes the operation of taking the symmetric
trace-free part.
\end{lemma}

\begin{proof}
As already discussed, Proposition \ref{Proposition:RegularityCondition} shows that the
property holds for exactly static initial data sets. In \cite{Fri98a} it has been
shown that this property concerns only the massless part of time symmetric
data sets. The result follows then by observing that an initial data set
which is static up to a certain order has the same massless part as the
one of a static initial data set so that $b_{ij}=\mathring{b}_{ij}$.
\end{proof}

\section{The cylinder at spatial infinity}
\label{Section:CylinderAtSpatialInfinity}

In \cite{Fri98a} a representation of the region of spacetime close to
null infinity and spatial infinity has been introduced ---see also the
comprehensive discussion in \cite{Fri04}. The standard representation
of this region of spacetime depicts $i^0$ as a point.  In contrast,
the representation introduced in \cite{Fri98a} depicts spatial
infinity as a cylinder ---\emph{the cylinder at spatial
infinity}. This construction is briefly reviewed for the case of time
symmetric initial data sets with an analytic conformal metric in a
neighbourhood $\mathcal{B}_a(i)$ of infinity. The reader is referred
to \cite{Fri98a,Fri04} for a thorough discussion of the details ---see also \cite{Val10}.

\subsection{The Manifold $\mathcal{C}_a$} \label{section:manifold_Ca} 

Starting from the initial hypersurface $\mathcal{S}$, the construction
introduced in \cite{Fri98a} makes use of a blow-up of the point $i\in
\mathcal{S}$ to the 2-sphere $\Sphere^2$. This blow-up requires the
introduction of a particular bundle of spin-frames over
$\mathcal{B}_a$. Consider the (conformally rescaled) spacetime
$(\mathcal{M},g_{\mu\nu})$ obtained as the development of the time
symmetric initial data set $(\mathcal{S},h_{ij})$. Let
$SL(\mathcal{S})$ be the set of spin dyads
$\delta=\{\delta_A\}_{A=0,1}$ on $\mathcal{S}$ which are normalised
with respect to the alternating spinor $\epsilon_{AB}$ in such a way
that $\epsilon_{01}=1$. Let $\tau=\sqrt{2}e_0$, where
$e_0$ is the future $g$-unit normal of $\mathcal{S}$ and $\tau_{AA'}$
its spinorial counterpart. The spinor $\tau_{AA'}$ enables the
introduction of space-spinors ---sometimes also called $SU(2)$
spinors, see \cite{Ash91,Fra98a,Som80}. It defines a sub-bundle
$SU(\mathcal{S})$ of $SL(\mathcal{S})$ with structure group
$SU(2,\Complex)$ and projection $\pi$ onto $\mathcal{S}$. Given a
spinorial dyad $\delta\in SU(\mathcal{S})$ one can define an
associated vector frame $e_a$, $a=1,2,3$.  We
shall restrict our attention to dyads related to frames
$\{e_j\}_{j=0,\cdots,3}$ on $\mathcal{B}_a$ such that $e_3$ is tangent
to the $h$-geodesics starting at $i$. Let $\check{H}$ denote the
horizontal vector field on $SU(\mathcal{S})$ projecting to the radial
vector $e_3$. The fibre $\pi^{-1}(i)\subset SU(\mathcal{S})$ (the
fibre ``over'' $i$) can be parametrised by choosing a fixed dyad
$\delta^*$ and then letting the group $SU(2,\Complex)$ act on it. Let
$(-a,a)\ni \rho \mapsto \delta(\rho,\updn{t}{A}{B})\in
SU(\mathcal{S})$ be the integral curve to the vector $\check{H}$
satisfying $\delta(0,\updn{t}{A}{B})=\delta(\updn{t}{A}{B})\in
\pi^{-1}(i)$. With this notation one defines the set
\[
\mathcal{C}_a =\big\{ \delta(\rho,\updn{t}{A}{B})\in SU(\mathcal{B}_{a}) \;\big|\; |\rho|<a, \; \updn{t}{A}{B}\in SU(2,\Complex)\big\},
\]
which is a smooth submanifold of $SU(\mathcal{S})$ diffeomorphic to
$(-a,a)\times SU(2,\Complex)$. It follows that the projection map
$\pi$ of the bundle $SU(\mathcal{S})$ maps $\mathcal{C}_a$ into
$\mathcal{B}_a$. The manifold $\mathcal{C}_a$ inherits a number of
structures from $\mathcal{B}_a$. In particular, the solder and
connection forms can be pulled back to smooth 1-forms on
$\mathcal{C}_a$ satisfying the \emph{structure equations} which relate
them to the \emph{curvature form}. The explicit form of the structure
equations will not be required here. As
$\mathcal{S}$ is 3-dimensional, the full content of its curvature can
be expressed in terms of the spinorial counterpart of the trace-free
Ricci tensor $s_{ABCD}$ and the Ricci scalar $r$.

\medskip
In the sequel $\updn{t}{A}{B}\in SU(2,\Complex)$ and $\rho\in \Real$
will be used as coordinates on $\mathcal{C}_a$. Consequently, one has
that $\check{H}=\partial_\rho$. Vector fields $X_\pm$, $X$ relative to
the $SU(2,\Complex)$-dependent part of the coordinates can be
introduced by requiring the commutation relations
\[
[X,X_+]=2X_+, \quad [X,X_-]=-2X_-, \quad [X_+,X_-]=-X,
\]
and by requiring that they commute with $\check{H}=\partial_\rho$.
More importantly, it can be seen that for
$p\in \mathcal{B}_a\setminus\{i\}$ the projections of the fields
$\partial_\rho$, $X_\pm$ span the tangent space at $p$. Given these
vector fields, define the frame $c_{AB}=c_{(AB)}$ by
\[
c_{AB}= c^1_{AB}\partial_\rho + c^+_{AB}X_+ + c^-_{AB}X_-,
\]
where
\[
c^1_{AB} = x_{AB}, \quad c^+_{AB} = \frac{1}{\rho} z_{AB} + \check{c}^+_{AB}, \quad c^-_{AB} = \frac{1}{\rho}y_{AB} + \check{c}^-_{AB},
\]
with constant spinors $x_{AB}$, $y_{AB}$ and $z_{AB}$ given by
\[
x_{AB}\equiv \sqrt{2} \dnup{\epsilon}{(A}{0} \dnup{\epsilon}{B)}{1}, \quad y_{AB} \equiv -\frac{1}{\sqrt{2}}\dnup{\epsilon}{A}{1} \dnup{\epsilon}{B}{1}, \quad z_{AB}=\frac{1}{\sqrt{2}} \dnup{\epsilon}{A}{0} \dnup{\epsilon}{B}{0},
\]
and analytic spinor fields satisfying
\[
\check{c}^\pm_{AB} = \mathcal{O}(\rho), \quad \check{c}^\pm_{01}=0.
\]
For the frame $c_{AB}$, the connection coefficients, $\gamma_{ABCD}$, are given by 
\[
\gamma_{ABCD} = \frac{1}{\rho} \gamma^*_{ABCD} + \check{\gamma}_{ABCD}, \quad \gamma^*_{ABCD} = \frac{1}{2} (\epsilon_{AC}x_{BD}+\epsilon_{BD}x_{AC}),
\]
where
\[
\check{\gamma}_{ABCD}=\mathcal{O}(\rho), \quad \check{\gamma}_{11CD}=0.
\]

\subsection{An orthonormal basis for functions on $SU(2,\Complex)$}
Given $\updn{t}{A}{B}\in SU(2,\Complex)$,
define
\begin{eqnarray*}
&& \TT{m}{j}{k}(\updn{t}{A}{B}) = \binom{m}{j}^{1/2} \binom{m}{k}^{1/2} \updn{t}{(B_1}{(A_1}\cdots \updn{t}{B_m)_j}{A_m)_k}, \\
&& \TT{0}{0}{0}(\updn{t}{A}{B})=1,
\end{eqnarray*}
with $j,k=0,\ldots,m$ and $m=1,2,3,\ldots$. The subindex
expression ${}_{(A_1\cdots A_m)_k}$ means that the indices are
symmetrised and then $k$ of them are set equal to $1$, while the
remaining ones are set to $0$. Details about the properties of these
functions can be found in \cite{Fri86a,Fri98a}.  The functions
$\sqrt{m+1}\TT{m}{j}{k}$ form a complete orthonormal set in the
Hilbert space $L^2(\mu,SU(2,\Complex))$, where $\mu$ denotes the
normalised \emph{Haar measure} on $SU(2,\Complex)$. 
The action of the differential operators $X_\pm$ on the
functions $\TT{m}{k}{j}$ is given by
\[
 X_+ \TT{m}{k}{j} = \sqrt{j(m-j+1)} \TT{m}{k}{j-1}, \quad X_-\TT{m}{k}{j}=-\sqrt{(j+1)(m-j)}\TT{m}{k}{j+1}.
\]
In the sequel, we will need to linearise products of the form
$\TT{m}{i}{k}$. To this end, we use the formula:
\begin{eqnarray}
\label{Clebsch-Gordan}
&& \TT{i_1}{j_1}{k_1}\times \TT{i_2}{j_2}{k_2} = \sum^\mu_{p=0} D(i_1,j_1,k_1;i_2,j_2,k_2;i_1+i_2-2p,j_1+j_2-p,k_1+k_2-p) \nonumber \\
&&  \hspace{4cm}\times\TT{i_1+i_2-2p}{j_1+j_2-p}{k_1+k_2-p},
\end{eqnarray}
with $\mu=\min\{i_1,i_2,j_1+j_2,k_1+k_2\}$ and 
\[
D(i_1,j_1,k_1;i_2,j_2,k_2; l,m,n)= C(i_1,j_1;i_2,j_2;l,m)\overline{C(i_1,k_1;i_2,k_2;l,m)},
\]
and $C(i,j;k,l;m,n)$ the Clebsch-Gordan coefficients of $SU(2,\Complex)$.

\subsection{Normal expansions at $\mathcal{I}^0$}
In the sequel, we will consider the \emph{lift} of analytic fields
defined on $\mathcal{B}_a$ to $\mathcal{C}_a$. In particular, the lift
of $|x|$ is $\rho$. More generally, let $\xi_{A_1B_1\cdots A_l B_l}$
denote a spinorial field on $\mathcal{B}_a$. Denote, again, by
$\xi_{A_1B_1\cdots A_l B_l}$ its lift to $\mathcal{C}_a$. Denote by
$\xi_j=\xi_{(A_1B_1\cdots A_l B_l)_j}$, $0\leq j \leq l$ its essential
components. The function $\xi_j$ has spin weight $s=l-j$ and a unique
expansion of the form
\[
\xi_j = \sum^\infty_{p=0} \xi_{j,p}\rho^p, \quad \xi_{j,p} = \sum_{q= \max\{|l-j|, l-p\}}^{p+l} \sum_{k=0}^{2q} \xi_{j,p;2q,k} \TT{2q}{k}{2q-l+j},
\]
with complex coefficients $\xi_{j,p;2q,k}$. More generally, we shall
consider symmetric spinorial fields $\xi_{A_1\cdots A_r}$ on
$\mathcal{C}_a$ with independent components $\xi_j = \xi_{(A_1\cdots
A_{2r})_j}$, $0\leq j \leq 2r$, and spin-weight $s=r-j$ which do not
descend to analytic spinor fields on $\mathcal{B}_a$. In this case one
has that
\[
\xi_j =\sum^\infty_{p=0} \xi_{j,p}\rho^p, \quad \xi_{j,p}= \sum_{q=|r-j|}^{q(p)} \sum_{k=0}^{2q} \xi_{j,p;2q,k} \TT{2q}{k}{q-r+j},
\]
where one has \emph{a priori} that $0\leq |r-j| \leq q(p)$. An
expansion of the latter form will be said to be of type $q(p)$. 

\subsubsection{Particular results concerning expansion types}
We recall the following result in \cite{Fri98a}:
\begin{lemma}
Assuming the cn-gauge and if for some non-negative integer $k$ it holds 
\[
D_{(E_pF_p}\cdots D_{E_1F_1} b_{ABCD)}(i)=0, \quad p=0,1,\cdots, k,
\]
then we have the following expansion types:
\begin{eqnarray*}
&& \mbox{\em type}(r)=\mbox{\em type}(W)=p, \\
&& \mbox{\em type}(s_{ABCD})=p, \\
&& \mbox{\em type}(\check{\gamma}_{ABCD})=p-1, \\
&& \mbox{\em type}(\check{c}^\pm_{AB})=p-1, \\
&& \mbox{\em type}(U-1) =p-2.
\end{eqnarray*}
In addition, we have that
\[
U=1 + O(\rho^4), \quad s_{ABCD}=\mathcal{O}(\rho^2), \quad \check{c}_{AB}^\pm =\mathcal{O}(\rho^3), \quad \check{\gamma}_{ABCD} =\mathcal{O}(\rho^3).
\]
\end{lemma}

\noindent
\textbf{Remark.} The particular structure of the expansions dictated
by this lemma greatly simplifies the subsequent analysis.

\medskip
Important for our subsequent analysis is the particular form of the
lift to $\mathcal{C}_{a}$ of the function $\breve{W}$ appearing
in the conformal factor $\vartheta$ for data which is static up to
order $\pb$. Using the methods of \cite{Fri98a} one finds that
\begin{equation}
\label{breveWLifted}
\breve{W} = \sum_{p=\pb+1}^{\pb+3} \sum_{k=0}^{2p} \frac{1}{p!}
\breve{w}_{p;2p,k} \TT{2p}{k}{p}\rho^p + \sum_{p=\pb+4}^\infty
\sum_{q=0}^{p} \sum_{k=0}^{2q} \frac{1}{p!} \breve{w}_{p;2q,k} \TT{2q}{k}{q}\rho^p,
\end{equation}
where the coefficients $\breve{w}_{p;2q,k}$ are related to the
constant tensors $w_{i_1\cdots i_p}$ via contractions with Infeld-van
der Waerden symbols. In particular, it can be seen that
\[
\breve{w}_{p;2q,k}=0, \quad q=0,\ldots p, \quad k=0,\ldots 2q \Longleftrightarrow w_{i_1\cdots i_p}=0.
\]

 The key observation in equation \eqref{breveWLifted} is that the
terms of order $\mathcal{O}(\rho^{\pb+1})$,
$\mathcal{O}(\rho^{\pb+2})$, $\mathcal{O}(\rho^{\pb+3})$ are formally
identical to the ones appearing in the corresponding expansions for
the function$\breve{W}$ of data which is Schwarzschildean
 up to order $\pb$ ($\mathring{W}=m/2$). In other words,
\begin{eqnarray*}
&& \breve{w}_{\pb+1;q,k}=0, \quad q=0,\ldots \pb+1, \quad k=0,\ldots q, \\
&& \breve{w}_{\pb+2;q,k}=0, \quad q=0,\ldots \pb+2, \quad k=0,\ldots q, \\
&& \breve{w}_{\pb+3;q,k}=0, \quad q=0,\ldots \pb+3, \quad k=0,\ldots q.
\end{eqnarray*}

\subsection{The spacetime Friedrich gauge}
\label{section:Fgauge}
 The formulation of the initial value problem near spatial infinity
presented in \cite{Fri98a} employs gauge conditions based on timelike
conformal geodesics. The conformal geodesics are curves which are
autoparallel with respect to a Weyl connection ---i.e. a torsion-free
connection which is not necessarily the Levi-Civita connection of a
metric. An analysis of Weyl connections in the context of the
conformal field equations has been given in \cite{Fri95}. In terms of
this gauge based on conformal geodesics ---which shall be called the
\emph{Friedrich gauge} or \emph{F-gauge} for short--- the conformal
factor of the spacetime can be determined explicitly in terms of the
initial data for the Einstein vacuum equations. Hence, provided that
the congruence of conformal geodesics and the fields describing the
gravitational field extend in a regular manner to null infinity, one
has complete control over the location of null infinity. This can be
ensured by making $\mathcal{B}_a$ suitably small. In addition, the
F-gauge renders a particularly simple representation of the
propagation equations. Using this framework, the singular initial
value problem at spatial infinity can be reformulated into another
problem where null infinity is represented by an explicitly known
hypersurface and where the data are regular at spacelike infinity.
The construction of the bundle manifold $\mathcal{C}_a$ and the
blowing up of the point $i\in \mathcal{B}_a$ to the set
$\mathcal{I}^0\subset \mathcal{C}_a$, briefly described in section
\ref{section:manifold_Ca}, are the first steps in the construction of
this regular setting. The next step is to introduce a rescaling of the
frame bundle so that fields that are singular at $\mathcal{I}^0$
become regular. 

\bigskip
Following the discussion of \cite{Fri98a} assume that given the
development of data prescribed on $\mathcal{B}_a$, the timelike
spinor $\tau^{AA'}$ introduced in section
\ref{section:manifold_Ca} is tangent to a congruence of timelike
conformal geodesics which are orthogonal to $\mathcal{B}_a$. The
canonical conformal factor rendered by this congruence
of conformal geodesics is given in terms of an affine parameter
$\tau$ of the conformal geodesics by
\begin{equation}
\label{Theta}
\Theta=\kappa^{-1}\Omega\left(1-\frac{\kappa^2\tau^2}{\omega^2}\right),
\end{equation}
with $\omega$ as given by equation \eqref{omega} and 
where $\Omega=\vartheta^{-2}$ and $\vartheta$ solves the Yamabe
equation (\ref{YamabeEquation}) ---see
\cite{Fri95,Fri98a,Fri03c}. The function $\kappa>0$  expresses the remaining
conformal freedom in the construction. It will be
taken to be of the form $\kappa=\kappa^\prime \rho$, with
$\kappa^\prime$ analytic, $\kappa'(i)=1$. Associated to the conformal
factor $\Theta$ there is a 1-form $d_\mu$ from which the Weyl
connection can be obtained. In spinorial terms, one has that for conformally flat data
\begin{equation}
d_{AA'}= \frac{1}{\sqrt{2}}\tau_{AA'} \partial_\tau{\Theta} -\updn{\tau}{B}{A'} d_{AB}, \quad
d_{AB}= 2\rho \left( \frac{Ux_{AB}-\rho D_{AB}U-\rho^2D_{AB}W}{(U+\rho W)^3} \right). \label{1-form}
\end{equation}
 The function $\kappa$ in the conformal factor $\Theta$,
induces a scaling $\delta_A \mapsto \kappa^{1/2} \delta_A$ of the spin
frame. Accordingly, one considers the bundle manifold
$\mathcal{C}_{a,\kappa}=\kappa^{1/2}\mathcal{C}_a$ of scaled spinor
frames. Using $\mathcal{C}_{a,\kappa}$ one defines the set
\[
\mathcal{M}_{a,\kappa}=\left\{ (\tau,q) \big |  q\in
\mathcal{C}_{a,\kappa}, -\frac{\omega(q)}{\kappa(q)} \leq \tau \leq \frac{\omega(q)}{\kappa(q)} \right\},
\]
which, assuming that the congruence of null
geodesics and the relevant fields extend adequately, can be identified with the
development of $\mathcal{B}_a$ up to null infinity ---that is, the region of
spacetime near null and spatial infinity. In addition, one defines the sets:
\begin{subequations}
\begin{eqnarray*}
&& \mathcal{I}=\big \{(\tau,q)\in \mathcal{M}_{a,\kappa} \;\big|\; \rho(q)=0, \;|\tau|<1\big\}, \\
&& \mathcal{I}^\pm= \big \{ (\tau,q)\in \mathcal{M}_{a,\kappa} \;\big |\; \rho(q)=0, \;\tau=\pm1\big \}, \\
&& \mathscr{I}^\pm=\left\{ (\tau,q)\in \mathcal{M}_{a,\kappa} \;\big | \; \rho(q)>0, \;\; \tau=\pm \frac{\omega(q)}{\kappa(q)} \right\},
\end{eqnarray*}
\end{subequations}
which will be referred to as, respectively, the \emph{cylinder at
spatial infinity}, the \emph{critical sets} and \emph{future and past
null infinity}. In order to coordinatise the hypersurfaces of constant
parameter $\tau$, one extends the coordinates $(\rho,\updn{t}{A}{B})$
off $\mathcal{C}_{a,\kappa}$ by requiring them to be constant along
the conformal geodesics ---i.e.  one has a system of \emph{conformal
Gaussian coordinates}. 

\bigskip
\noindent
\textbf{Remark.} For the purpose of the analysis carried out in this
article it turns out that the most convenient choice of the function
$\kappa$ in the conformal factor $\Theta$ of equation (\ref{Theta}) is
\[
\kappa =\rho.
\]
This leads to considerable simplifications in all the relevant
expressions. From this point onwards, this choice will always be
assumed.

\subsection{The evolution equations}
On the manifold $\mathcal{M}_{a,\kappa}$ it is possible to introduce a
calculus based on the derivatives $\partial_\tau$ and $\partial_\rho$
and on the operators $X_+$, $X_-$ and $X$. The operators
$\partial_\rho$, $X_+$, $X_-$ and $X$ originally defined on
$\mathcal{C}_{a}$ can be suitably extended to the rest of the manifold
by requiring them to commute with the vector field $\partial_\tau$. In
order to derive the propagation equations, a frame $c_{AA'}$ and the
associated spin connection coefficients $\Gamma_{AA'BC}$ of the Weyl
connection $\nabla$ will be used. The gravitational field is, in
addition, described by the spinorial counterparts of the Schouten
tensor of the Weyl connection, $\Theta_{AA'BB'}$, and of the rescaled
Weyl tensor, $\phi_{ABCD}$ ---see \cite{Fri95,Fri98a,Fri04}. Let
$\phi_i \equiv \phi_{(ABCD)_i}.$ In the present gauge, the information
of the spacetime spinors $c^\mu_{AA'}$, $\Theta_{AA'BB'}$ and
$\Gamma_{AA'BC}$ is encoded, respectively, in space spinors
$c^\mu_{AB}$, $\Theta_{ABCD}$ and $\Gamma_{ABCD}$ ---see \cite{Fri04}
for the detailed relation between the two sets of spinors.

In what follows, we will arrange the independent components of the
spinorial fields $c^\mu_{AB}$, $\Theta_{ABCD}$ and $\Gamma_{ABCD}$ in
the vector $\bm \upsilon$, and those of $\phi_{ABCD}$ in the vector $\bm
\phi$.  Suitable field equations for the fields contained in
${\bm\upsilon}$ and ${\bm\phi}$ can be obtained from the first and
second Cartan structure equations, the Bianchi identity of the
conformally rescaled spacetime $(\mathcal{M},g_{\mu\nu})$, and the
Bianchi identity of the physical spacetime
$(\tilde{\mathcal{M}},\tilde{g}_{\mu\nu})$ ---see
e.g. \cite{Fri03a,Fri04} for details. A solution to the equations thus
constructed implies a solution to the vacuum Einstein field
equations. The vector ${\bm \upsilon}$ has 45 independent complex components,
while ${\bm \phi}$ has 5 independent complex components.
Using the F-gauge it can be shown that the extended conformal field
equations given in \cite{Fri98a} imply the following evolution
equations for the unknowns ${\bm \upsilon}$
\begin{equation}
\label{upsilon:propagation}
\partial_\tau {\bm\upsilon} = \mathbf{K}\cdot{\bm\upsilon} + \mathbf{Q}({\bm\upsilon},{\bm\upsilon})+ \mathbf{L}\cdot{\bm\phi},
\end{equation}
where $\mathbf{K}$ and $\mathbf{Q}$ denote, respectively,  linear and
 quadratic constant matrix-valued functions with constant entries, and
$\mathbf{L}$ is a linear matrix-valued function with coefficients
depending on the coordinates and such that
$\mathbf{L}|_{\rho=0}=0$. For the unknowns ${\bm\phi}$, the Bianchi
identity $\nabla^{AA'}\phi_{ABCD}=0$ implies, respectively, a set of
propagation and constraint equations of the form:
\begin{subequations}
\begin{eqnarray}
&&\mathbf{E}\cdot \partial_\tau {\bm\phi} + \mathbf{A}^\mu (\mathbf{c})\partial_\mu {\bm\phi}
=\mathbf{F}({\bm\Gamma})\cdot {\bm\phi}, \label{bianchi:propagation}\\
&&\mathbf{B}^\mu (\mathbf{c}) \cdot\partial_{\mu}{\bm\phi} =\mathbf{G}({\bm\Gamma}),
\label{bianchi:constraint}
\end{eqnarray}
\end{subequations}
where $\mathbf{E}$ denotes the $5\times 5$ unit matrix,
$\mathbf{A}^\mu (\mathbf{c})$ and $\mathbf{B}^\mu (\mathbf{c})$,
$\mu=0,\ldots,3$, are, respectively, $5\times 5$ and $3\times 5$
constant matrix-valued linear functions depending on the frame field
coefficients $c^\mu_{AB}$. Finally, $\mathbf{F}({\bm\Gamma})$ and
$\mathbf{G}({\bm\Gamma})$ denote constant matrix-valued linear
functions of the connection coefficients $\Gamma_{ABCD}$.
 
\subsection{Initial data for the evolution equations}
\label{Section:InitialDataFormulae}
For quick reference, we include here the expressions for the initial data
for the  curvature components appearing in equations
\eqref{upsilon:propagation}-\eqref{bianchi:constraint}. These
expressions can be deduced from the conformal constraint equations
---see \cite{Fri98a}.  On $\mathcal{C}_{a,\kappa}$ one has that:
\begin{subequations}
\begin{eqnarray}
&& \Theta_{ABCD} = - \frac{\kappa^2}{\Omega} D_{(AB} D_{CD)} \Omega +
\frac{1}{12}\kappa^2 r h_{ABCD}, \label{Ricci:Data}\\ 
&& \phi_{ABCD} = \frac{\kappa^3}{\Omega^2}\left( D_{(AB} D_{CD)} +
  \Omega s_{ABCD} \right), \label{Weyl:Data}
\end{eqnarray}
\end{subequations}
where, $\Omega$ denotes the lift to $\mathcal{C}_{a,\kappa}$ of the
conformal factor obtained from equation \eqref{theta:perturbed}. The
expressions for the frame and connection coefficients do not involve
the function $\Omega$, and thus their detailed expressions
will not be required here.

\subsection{The transport equations} 
Consider now the system
(\ref{upsilon:propagation})-(\ref{bianchi:propagation}) with data
given on $\mathcal{C}_{a,\kappa}$, and let be given a neighbourhood $\mathcal{W}$
of $\mathcal{C}_{a,\kappa}$ in $\mathcal{M}_{a,\kappa}$ on which a
unique smooth solution of the Cauchy problem exists. From the point
of view of the propagation equations, the subset $\mathcal{W}\cap
\mathcal{I}$ is a regular hypersurface. Introduce the notation
\[
{\bm\upsilon}^{(0)}\equiv {\bm \upsilon}|_{\mathcal{W}\cap \mathcal{I}}, \quad 
{\bm \phi}^{(0)} \equiv {\bm \phi}|_{\mathcal{W}\cap \mathcal{I}}.
\]
 Due to the
property $\mathbf{L}|_{\rho=0}=0$, equations (\ref{upsilon:propagation})
decouple from equations (\ref{bianchi:propagation}) and can be
integrated on $\mathcal{W}\cap \mathcal{I}$ using the observation that
the restriction of the initial data to $\mathcal{I}^0$ coincides with
Minkowski data. The solutions thus obtained extend analytically to the
whole of $\mathcal{I}$ and in particular to the critical sets
$\mathcal{I}^\pm$. The set $\mathcal{I}$ turns out to be a \emph{total
characteristic} of the system
(\ref{upsilon:propagation})-(\ref{bianchi:propagation}) in the sense
that the whole system reduces to an interior system on $\mathcal{I}$.
Moreover, the constraint equations (\ref{bianchi:constraint}) also
reduce to an interior system on $\mathcal{I}$. As mentioned before,
this feature is a consequence of the fact that the unphysical metric
$g_{\mu\nu}$ determined by a solution to the conformal field equations
degenerates as $\rho\rightarrow 0$.

\medskip
A crucial
structural property is that
\begin{equation}
\label{degeneracy}
\mathbf{E} + \mathbf{A}^0(\mathbf{c})= \mbox{diag}(1+\tau,1,1,1,1-\tau) \quad \mbox{ on } \mathcal{I},
\end{equation}
so that the matrix $\mathbf{E} + \mathbf{A}^0(\mathbf{c})$ which is
positive definite for $|\tau|<1$ degenerates at
$\mathcal{I}^\pm$. Understanding the effects of this degeneracy is the
main motivation behind the analysis in the present article and that in
\cite{Val10}.

\medskip
In order to analyse the effects of the degeneracy of the matrix
$\mathbf{E} + \mathbf{A}^0(\mathbf{c})$ we will consider the hierarchy
of transport equations which is obtained by repeated application of
the operator $\partial_\rho$ to equations
\eqref{upsilon:propagation}-\eqref{bianchi:constraint} and then
evaluation on $\mathcal{I}$. By this procedure one obtains interior
systems for the quantities
\[
{\bm\upsilon}^{(p)}\equiv \partial^{(p)}_\rho {\bm\upsilon}|_{\mathcal{I}}, \quad {\bm\phi}^{(p)} \equiv \partial^{(p)}_\rho{\bm\phi}|_{\mathcal{I}}. 
\]
The transport equations take the following form for $p\geq 1$:
\begin{subequations}
\begin{eqnarray}
&&\hspace{-8mm}\partial_\tau {\bm\upsilon}^{(p)} = \mathbf{K}\cdot{\bm\upsilon}^{(p)}+\sum_{j=0}^{p}\binom{p}{j}\left(\mathbf{Q}({\bm\upsilon}^{(j)},{\bm\upsilon}^{(p-j)})+ \mathbf{L}^{(j)}\cdot{\bm\phi}^{(p-j)}\right), \label{upsilon:transport} \\
&&\hspace{-8mm}\mathbf{E}\cdot\partial_\tau{\bm\phi}^{(p)} +
(\mathbf{A}^\mu)^{(0)}\cdot\partial_\mu{\bm\phi}^{(p)}= \mathbf{F}({\bm\Gamma}^{(0)})\cdot{\bm\phi}^{(p)} \nonumber \\
&& \hspace{4.2cm} +\sum_{j=1}^p
\binom{p}{j}\left(\mathbf{F}({\bm\Gamma}^{(j)})\cdot{\bm\phi}^{(p-j)}-(\mathbf{A}^\mu)^{(j)}\cdot\partial_\mu
 {\bm\phi}^{(p-j)}\right), \label{pbianchi:transport} \\
&&\hspace{-8mm} (\mathbf{B}^\mu)^{(0)}\cdot\partial_\mu{\bm\phi}^{(p)} =\mathbf{G}({\bm\Gamma}^{(0)})\cdot \phi^{(p)}  +\sum_{j=1}^p\binom{p}{j}\left(\mathbf{G}({\bm\Gamma}^{(j)})\cdot{\bm\phi}^{(p-j)}-(\mathbf{B}^\mu)^{(j)}\cdot\partial_\mu {\bm\phi}^{(p-j)}\right). \label{cbianchi:transport}
\end{eqnarray}
\end{subequations}
Note that the non-homogeneous terms in the equations
(\ref{upsilon:transport})-(\ref{cbianchi:transport}) depend on
${\bm\upsilon}^{(p')}$, ${\bm\phi}^{(p')}$ for $0\leq p' <p$. Thus, if
their values are known, then equations 
(\ref{upsilon:transport})-(\ref{pbianchi:transport}) constitutes an
interior system of linear equations for ${\bm\upsilon}^{(p)}$ and
${\bm\phi}^{(p)}$. The principal part of these equations is universal,
in the sense that it is independent of the value of $p$. If the
initial data on $\mathcal{C}_{a,\kappa}$ for the system
(\ref{upsilon:propagation})-(\ref{bianchi:propagation}) is analytic
---as it is the case in the present analysis--- then suitable initial
data for the interior system
(\ref{upsilon:transport})-(\ref{pbianchi:transport}) can be obtained
by repeated $\rho$-differentiation and evaluation on $\mathcal{I}^0$.

\bigskip
The language of jets is natural in the present context. For
$p=0,1,2,\ldots$ and any sufficiently smooth (possibly vector valued)
function $f$ defined on $\mathcal{M}_{a,\kappa}$, the sets of
functions $\{f^{(0)}$, $f^{(1)}$, \ldots, $f^{(p)}\}$ on $\mathcal{I}$
will be denoted by $J^{(p)}_\mathcal{I}[f]$ and referred to as
\emph{the jet order $p$ of $f$ on} $\mathcal{I}$ ---and similarly with
$\mathcal{I}$ replaced by $\mathcal{I}^0$. If $\mathbf{u}=({\bm\upsilon},{\bm\phi})$ is a
solution to the equations (\ref{upsilon:transport}),
(\ref{pbianchi:transport}) and (\ref{cbianchi:transport}), we refer to
$J^{(p)}_\mathcal{I}[\mathbf{u}]$ as to \emph{the s-jet of $\mathbf{u}$ of order $p$}
and to the data $J^{(p)}_{\mathcal{I}^0}[\mathbf{u}]$ as to \emph{the d-jet of
$\mathbf{u}$ of order $p$}. An s-jet $J^{(p)}_\mathcal{I}[\mathbf{u}]$ of order $p$ will
be called \emph{regular} on $ \overline{\mathcal{I}}\equiv \mathcal{I} \cup
\mathcal{I}^+ \cup \mathcal{I}^-$ if the corresponding functions
extend smoothly to the critical sets $\mathcal{I}^\pm$.

\medskip
We note the following  result of \cite{Val10} which will be used in the sequel.
\begin{lemma}
\label{Lemma:LogsLeadToLogs}
  If the s-jets $J^{(p-1)}_{\mathcal{I}}[{\bm\upsilon}]$ and
  $J^{(p-1)}_{\mathcal{I}}[{\bm\phi}]$ have polynomial dependence in $\tau$
  for some $p\geq 1$, then $J^{(p)}_{\mathcal{I}}[{\bm\upsilon}]$ has also
  polynomial dependence in $\tau$.
\end{lemma}

\subsection{Decomposition in spherical harmonics}

Our analysis requires decomposing the entries of the vector unknowns
${\bm\upsilon}^{(p)}$ and ${\bm\phi}^{(p)}$ in terms of the functions
$\TT{i}{j}{k}$. Following the discussion in \cite{Val10}, given a
vector $\mathbf{u}^{(p)}=(u_1^{(p)},\ldots,u^{(p)}_N)$ with entries in
$\Real\times \Sphere^2$ and non-negative integers $q$ and
$k=0,\ldots,2q$ one defines the \emph{sector}
$\mathfrak{S}_{q,k}[\mathbf{u}^{(p)}]$ as the collection of coefficients
\[
u_{i,p;2q,k} \equiv (2q+1)\int_{SU(2)} \bar{u}_i^{(p)} \TT{2q}{k}{q-s}\mbox{d}\mu,
\]
where $s$ is the spin-weight of $u^{(p)}_i$, $\mbox{d}\mu$ is the
Haar measure of $SU(2)$ and $\bar{u}_i^{(p)}$ denotes the
complex conjugate  of $u_i^{(p)}$. Furthermore, one defines
\[
\mathfrak{S}_q[\mathbf{u}^{(p)}] \equiv \bigcup^{2q}_{k=0} \mathfrak{S}_{q,k}[\mathbf{u}^{(p)}]. 
\]
With this notation, a sector will be said to vanish if
$\mathfrak{S}_q[\mathbf{u}^{(p)}]$. For convenience in the discussion,
we introduce the following terminology:
\begin{itemize}
\item A coefficient $u_{i,p;2q,k}$ will be said \emph{static} if it
depends only on the mass $m$ and on the terms in the expansions of
$\mathring{U}$ and $\mathring{W}$. We will write
$u_{i,p;2q,k}=\mathring{u}_{i,p;2q,k}$.

\item The coefficient $u_{i,p;2q,k}$ will
be said to be \emph{Schwarzschildean} if it only depends on $m$. In
our gauge, Schwarzschildean terms can only arise in the sectors
$\mathcal{S}_0$.

\item A coefficient $u_{i,p;2q,k}$ will be said to be
\emph{non-static} if it depends on the terms in the
expansion of $\breve{W}$. We will write
$u_{i,p;2q,k}=\breve{u}_{i,p;2q,k}$.

\item A non-static coefficient $\breve{u}_{i,p;2q,k}$ will be said to
be a \emph{deviation from Schwarzschild} if it only depends on $m$ and
on the
coefficients in the leading term of $\breve{W}$
---i.e. $\breve{w}_{\pb+1,2(\pb+1),k}$, $k=0,\ldots,2(\pb+1)$. 

\end{itemize}

\medskip
The structure of the transport equations for the elements of the
various sectors $\mathfrak{S}_{q,k}[{\bm\upsilon}^{(p)}]$ and
$\mathfrak{S}_{q,k}[{\bm\phi}^{(p)}]$ has been discussed in detail
elsewhere ---see \cite{Fri98a,Fri04,Val10}. In particular, the
degeneracy at the critical sets $\mathcal{I}^\pm$ implied by
expression \eqref{degeneracy} is only relevant for sectors with $p\geq
2$. Thus, in the sequel it will always be assumed that $p\geq 2$. If
we denote by $'$ differentiation with respect to $\tau$, the
relevant transport equations are of the form:
\begin{eqnarray*}
&& {\bm\upsilon}'_{p;2q,k} = \mathbf{K} \cdot {\bm\upsilon}_{p;2q,k} + 2\mathbf{Q}({\bm\upsilon}_{0;0,0},{\bm\upsilon}_{p;2q,k}) + \mathbf{h}_{p;2q,k}, \\
&& (\mathbf{E}+\mathbf{A})\cdot {\bm\phi}'_{p;2q,k} +
\mathbf{A}_{p;2q,k}\cdot {\bm\phi}_{p;2q,k}=\mathbf{f}_{p;2q,k}, \\
&& \mathbf{B} \cdot {\bm\phi}'_{p;2q,k} + \mathbf{B}_{p;2q,k} \cdot {\bm\phi}_{p;2q,k}=\mathbf{g}_{p;2q,k},
\end{eqnarray*}
where
\begin{subequations}
\begin{eqnarray}
&& \hspace{-1.5cm}\mathbf{A} \equiv
\left(
\begin{array}{ccccc}
 1+\tau & 0 & 0 & 0 & 0  \\
0 & 1 & 0 & 0 & 0 \\
0 & 0 & 1 & 0 & 0 \\
0 & 0 & 0 & 1 & 0 \\
0 & 0 & 0 & 0 & 1-\tau
\end{array}
\right),
\quad 
\mathbf{A}_{p;2q} \equiv
\left(
\begin{array}{ccccc}
-(p-2) & \displaystyle \frac{1}{4}\beta_1 & 0 & 0 & 0 \\
-2\beta_1 & 1 & \displaystyle \frac{1}{3} \beta_2 & 0 & 0 \\
0 & \displaystyle \frac{3}{4}\beta_2 & 0 & \displaystyle
\frac{3}{4}\beta_2 & 0 \\
0 & 0 & -\displaystyle\frac{1}{3}\beta_2 & -1 & 2\beta_1 \\
0 & 0 & 0 & -\displaystyle\frac{1}{4}\beta_1 & p-2
\end{array}
\right),
\label{matrixA}\\
&& 
\hspace{-1.5cm}\mathbf{B} \equiv
\left(
\begin{array}{ccccc}
0& \tau & 0 & 0 & 0 \\
0& 0 & \tau & 0  & 0 \\
0 & 0 & 0 & \tau & 0
\end{array}
\right),
\quad 
\mathbf{B}_{p;2q} \equiv 
\left(
\begin{array}{ccccc}
2\beta_1 & -p & \displaystyle \frac{1}{3} \beta_2 & 0 & 0 \\
0 & \displaystyle\frac{3}{4}\beta_2 & -p
&\displaystyle\frac{3}{4}\beta_2 & 0 \\
0 & 0 & \displaystyle \frac{1}{3} \beta_2 & -p & 2\beta_1
\end{array}
\right) \label{matrixB}
\end{eqnarray}
\end{subequations}
with
\[
\beta_1 \equiv \sqrt{(q-1)(q+2)},\quad \beta_2\equiv\sqrt{q(q+1)}.
\]
The terms $\mathbf{f}_{p;2q,k}$, $\mathbf{g}_{p;2q,k}$,
$\mathbf{h}_{p;2q,k}$ are calculated from the left hand sides of
equations \eqref{upsilon:transport}-\eqref{cbianchi:transport} using
formula \eqref{Clebsch-Gordan}. Detailed expressions for
certain values of the multiindex $(p;q,k)$ will be given in the
sequel.

\section{The cylinder at spatial infinity for static spacetimes}
\label{Section:StaticCylinder}

As mentioned in the introduction, static initial data sets are
expected to play a privileged role among the class of time symmetric
initial data sets with a development which is asymptotically
simple. This point of view is supported by the following result proved
in \cite{Fri04} showing that the construction of the cylinder at
spatial infinity described in the previous section is for static
spacetimes as smooth as it is to be expected.

\begin{theorem}
For static vacuum solutions which are asymptotically flat, the
construction of the cylinder at spatial infinity is analytic in the
sense that there exists an $a>0$ and a frame for which all the fields,
including the rescaled conformal Weyl tensor extend to analytic fields
on some neighbourhood $\mathcal{N}$ of $\mathcal{I}$ in
$\mathcal{M}_{\kappa,a}$. This statement depends neither on the
particular choice of the conformal gauge used to prescribe the
(analytic) free datum on $\mathcal{S}$ nor on the choice of $\kappa$.
\end{theorem}

Essential for our analysis is the following corollary:

 %---cfr. Lemma 7.4 in \cite{Fri04}.

\begin{corollary} \label{StaticCorollary}
The solutions to the transport equations
\eqref{upsilon:transport}-\eqref{cbianchi:transport} for time
symmetric static data given in the cn-gauge extend analytically
through $\mathcal{I}^\pm$ for all orders $p$. Moreover, the solutions
to the transport equations are polynomial in $\tau$.
\end{corollary}

\begin{proof}
Once analyticity has been asserted, the polynomial
dependence of the solutions with respect to $\tau$ follows from an
analysis of the possible solutions to the reduced equations
---cfr. the discussion in section 6.2 and 6.3 of \cite{Val10}.
\end{proof}

\section{Properties of data sets which are static up to a given order}
\label{Section:DataStaticUpToAnOrder}

In what follows we will discuss some properties of initial data sets
which are static up to a certain order ---in the sense of definition
\ref{StaticUpToCertainOrder}. These properties have mainly to do with
the multipole structure of quantities associated to these initial data sets. 

\medskip
Given a quantity $u$, we will systematically write it as 
\[
u = \mathring{u} + \breve{u},
\]
where, in the terminology of the previous section, $\mathring{u}$ and
$\breve{u}$ denote, respectively, its static and non-static parts. For
quantities on the initial hypersurface $\mathcal{C}_{a,\kappa}$ we
similarly write $u_0= \mathring{u}_0 + \breve{u}_0$.

\medskip
Let as before, $\mathbf{c}$ and ${\bm\Gamma}$ denote, respectively,
the parts of the vectorial unknown ${\bm \upsilon}$ containing the
independent components of the frame and connection coefficients. As
mentioned in section \ref{Section:InitialDataFormulae}, the
expressions for the value of the frame and connection coefficients on
the initial hypersurface $\mathcal{C}_{a,\kappa}$ are independent of
$\Omega$. Thus, one has that:

\begin{lemma}
\label{LemmaStaticData}
For data which is static up to order $\pb\geq 0$ one has that
\[
\mathbf{c}_0^{(p)}=\mathring{\mathbf{c}}_0^{(p)}, \quad {\bm \Gamma}_0^{(p)}=\mathring{\bm \Gamma}_0^{(p)},
\]
for all $0 \leq p \leq \pb$.
\end{lemma}
Crucial for our analysis is the behaviour of the data for the curvature spinors
$\Theta_{ABCD}$ and $\phi_{ABCD}$. It is observed that
\begin{subequations}
\begin{eqnarray}
&& \Omega = \rho^2 -m\rho^3 +\frac{3}{4}m^2+ \mathcal{O}(\rho^5),
\label{Omega1}\\ 
&& D_{(AB}D_{CD)}\Omega = -6m \epsilon^2_{ABCD} +
12m^2 \epsilon^2_{ABCD} \rho^2 +\mathcal{O}(\rho^3), \label{Omega2}
\end{eqnarray}
\end{subequations}
from where it follows that the leading terms in these expressions are
Schwarzschildean. The first non-Schwarzschildean contributions arise
at orders $\mathcal{O}(\rho^5)$ and $\mathcal{O}(\rho^3)$, respectively ---see
e.g. \cite{FriKan00}. Using the expansions \eqref{Omega1}-\eqref{Omega2} together with formulae \eqref{Ricci:Data}-\eqref{Weyl:Data} one obtains after a long but straightforward computation the following:

\begin{lemma}
\label{LemmaQuasiStaticData}
For data which is static up to order $\pb\geq 2$ one has that the
curvature fields on $\mathcal{C}_{a,\kappa}$ satisfy
\begin{eqnarray*}
&& \breve{\Theta}_{ABCD}\equiv \Theta_{ABCD}-\mathring{\Theta}_{ABCD} = \mathcal{O}(\rho^{\pb+2}), \\
&& \breve{\phi}_{ABCD} \equiv \phi_{ABCD}-\mathring{\phi}_{ABCD}=\mathcal{O}(\rho^{\pb+1}).
\end{eqnarray*}
Moreover, the terms
\begin{eqnarray*}
&& \breve{\Theta}^{(\pb+1)}_{ABCD}, \quad \breve{\Theta}^{(\pb+2)}_{ABCD}, \quad \breve{\Theta}^{(\pb+3)}_{ABCD}, \quad\breve{\Theta}^{(\pb+4)}_{ABCD}, \\
&& \breve{\phi}^{(\pb+1)}_{ABCD}, \quad \breve{\phi}^{(\pb+2)}_{ABCD}, \quad \breve{\phi}^{(\pb+3)}_{ABCD},
\end{eqnarray*}
on $\mathcal{C}_{a,\kappa}$ are of the form of deviations from 
Schwarzschild data. On the other hand,
$\breve{\phi}^{(\pb+4)}_{ABCD}$ contains more general types of deviations.
\end{lemma}

\medskip
An inspection of the explicit computations in \cite{FriKan00} one
obtains the following  result.

\begin{lemma}
\label{Lemma:MultipoleStructureStatic}
For a data which is static up to order $\pb\geq 3$ one has that 
\begin{eqnarray*}
&& {\bm\upsilon}^{(0)}=\mathring{\bm\upsilon}^{(0)}, \quad
{\bm \upsilon}^{(1)}=\mathring{\bm \upsilon}^{(1)}, \quad
{\bm \upsilon}^{(2)}=\mathring{\bm \upsilon}^{(2)}, \\
&& \mathbf{L}^{(1)} =\mathring{\mathbf{L}}^{(1)},
\quad \mathbf{L}^{(2)} = \mathring{\mathbf{L}}^{(2)},  \quad
\mathbf{L}^{(3)}=\mathring{\mathbf{L}}^{(3)},\\
&& {\bm\phi}^{(0)}=\mathring{\bm \phi}^{(0)}, \quad {\bm\phi}^{(1)}=\mathring{\bm \phi}^{(1)}
\end{eqnarray*}
have only contributions to the sector $\mathfrak{S}_0$. On the other
hand, the coefficients of 
\[
{\bm\upsilon}^{(3)}=\mathring{\bm\upsilon}^{(3)}, \quad {\bm\phi}^{(2)}=\mathring{\bm\phi}^{(2)}, 
\]
contain contributions to the sectors $\mathfrak{S}_0$ and
$\mathfrak{S}_2$. 
\end{lemma}

\noindent
\textbf{Remark.} In the previous lemma $\mathbf{L}$ denotes the
linear matrix-valued function appearing in equation
\eqref{upsilon:propagation}. It depends on the conformal factor
$\Theta$ as given by \eqref{Theta} and on the 1-form $d_{AB}$ as given
by \eqref{1-form}.

\medskip
A lenghty but straightforward computation using formulae \eqref{Theta}
and \eqref{1-form} renders the following lemma.

\begin{lemma}
\label{Lemma:HigherOrderL}
For data which is static up to order $\pb\geq 3$ one has
that 
\[
\breve{\mathbf{L}}^{(\pb+1)}=0, \quad \breve{\mathbf{L}}^{(\pb+2)}=0.
\]
\end{lemma}

\section{Solutions to the transport equations for data which is static
 up to a certain order}
\label{Section:CoreAnalysis}

In this section we consider a time symmetric initial data set which is
static up to order $p=\pb+1$. The data can be parametrised in the form
\[
{\bm\upsilon}_0 = \mathring{\bm\upsilon}_0 + \breve{\bm\upsilon}_0, \quad {\bm\phi}_0 = \mathring{\bm\phi}_0 + \breve{\bm\phi}_0
\]
where $\mathring{\bm\upsilon}_0$, $\mathring{\bm\phi}_0$ corresponds to
exactly static initial data, while $\breve{\bm\upsilon}_0=0$, $\breve{\bm\phi}_0=0$ if
$W=m/2$.  A similar splitting can be applied to the quantities
${\bm\upsilon}^{(p)}$, ${\bm\phi}^{(p)}$. More precisely,
\[
{\bm\upsilon}^{(p)} = \mathring{\bm\upsilon}^{(p)} + \breve{\bm\upsilon}^{(p)}, \quad {\bm\phi}^{(p)} =
\mathring{\bm\phi}^{(p)}+ \breve{\bm\phi}^{(p)}, \quad p\geq 0.
\]

One has the following result.

\begin{lemma}
\label{Lemma:LowerOrderSolutions}
For initial data which is static up to order $p=\pb$, one has that
\[
\breve{\bm \upsilon}^{(p)} =0, \quad  0\leq p \leq \pb+1,
\]
and 
\[
\breve{\bm\phi}^{(p)} =0, \quad 0\leq p \leq \pb.
\]
\end{lemma}

In other words, the first deviations from static evolution appear in
${\bm \phi}^{(\pb+1)}$. 

\medskip
\noindent
\begin{proof}  One first notes that 
\[
\breve{\bm\upsilon}^{(p)}_0 =0, \quad p\leq \pb+1
\]
and 
\[
\breve{\bm\phi}^{(p)}_0 =0, \quad p\leq \pb. 
\]
One then argues by induction. It is clear that
\[
\breve{\bm\upsilon}^{(0)}=0, \quad \breve{\bm\phi}^{(0)}=0.
\]
Now, given $p$ such that $1\leq p \leq \pb$, assume that
$\breve{\bm\upsilon}^{(p')}=0$ and $\breve{\bm\phi}^{(p')}=0$ for $0\leq
p'<p$. Substitution of the Ans\"atze
${\bm \upsilon}^{(j)}=\mathring{\bm\upsilon}^{(j)}+\breve{\bm\upsilon}^{(j)}$,
${\bm\phi}^{(j)}=\mathring{\bm\phi}^{(j)}+\breve{\bm\phi}^{(j)}$, $j=0,\ldots,p$ into the
$\upsilon$-transport equation \eqref{upsilon:transport} renders the equation
\begin{eqnarray*}
&& \partial_\tau \mathring{\bm \upsilon}^{(p)} + \partial_\tau
\breve{\bm \upsilon}^{(p)} = \mathbf{K}\cdot\mathring{\bm\upsilon}^{(p )}+  \mathbf{K}\cdot\breve{\bm \upsilon}^{(p)} 
\\
&& \hspace{4cm}+\sum_{j=0}^{p}\binom{p}{j}\left(\mathbf{Q}(\mathring{\bm \upsilon}^{(j)},\mathring{\bm \upsilon}^{(p-j)})+
  \mathring{\mathbf{L}}^{(j)}\cdot\mathring{\bm \phi}^{(p-j)}\right)  \\
&&
\hspace{4cm}+\sum_{j=0}^{p}\binom{p}{j}\left(\mathbf{Q}(\mathring{\bm \upsilon}^{(j)},\breve{\bm\upsilon}^{(p-j)})+
  \mathring{\mathbf{L}}^{(j)}\cdot\breve{\bm\phi}^{(p-j)}\right)  \\
&& \hspace{4cm}+\sum_{j=0}^{p}\binom{p}{j}\left(\mathbf{Q}(\breve{\bm \upsilon}^{(j)},\mathring{\bm \upsilon}^{(p-j)})+
  \breve{\mathbf{L}}^{(j)}\cdot\mathring{\bm \phi}^{(p-j)}\right),  \\
&&  \hspace{4cm}+\sum_{j=0}^{p}\binom{p}{j}\left(\mathbf{Q}(\breve{\bm \upsilon}^{(j)},\breve{\bm \upsilon}^{(p-j)})+
  \breve{\mathbf{L}}^{(j)}\cdot\breve{\bm \phi}^{(p-j)}\right). 
\end{eqnarray*}
Using the induction hypothesis and
\[
 \partial_\tau \mathring{\bm \upsilon}^{(p)} = \mathbf{K}\cdot\mathring{\bm \upsilon}^{(p )}+\sum_{j=0}^{p}\binom{p}{j}\left(\mathbf{Q}(\mathring{\bm \upsilon}^{(j)},\mathring{\bm \upsilon}^{(p-j)})+
  \mathring{\mathbf{L}}^{(j)}\cdot\mathring{\bm \phi}^{(p-j)}\right),
\]
one finds that 
\[
\partial_\tau
\breve{\bm \upsilon}^{(p)} = \mathbf{K}\cdot\breve{\bm \upsilon}^{(p)} +
\mathbf{Q}(\mathring{\bm \upsilon}^{(0)},\breve{\bm \upsilon}^{(p)}) +
\mathbf{Q}(\breve{\bm\upsilon}^{(p)},\mathring{\bm \upsilon}^{(0)})
-\mathring{\mathbf{L}}^{(0)}\cdot\breve{\bm\phi}^{(p)} - \breve{\mathbf{L}}^{(p)}
\cdot\mathring{\phi}^{(0)}.  
\]
As this last equation is homogeneous in the unknown $\breve{\bm \upsilon}^{(p)}$,
the initial condition $\breve{\bm \upsilon}^{(p)}_0=0$ implies
$\breve{\bm \upsilon}^{(p)}=0$. A similar argument leads to the following
equation for $\breve{\bm \phi}^{(p)}$:
\[
\sqrt{2}\mathbf{E}\cdot \partial_\tau \breve{\bm \phi}^{(p)} + (\mathring{\mathbf{A}}^\mu)^{(0)} \cdot\partial_\mu
\breve{\bm \phi}^{(p)} = \mathbf{F}^{(0)} \cdot \breve{\bm\phi}^{(p)}. 
\]
Again, the homogeneity of this last equation, together with the
initial condition $\breve{\bm\phi}^{(p)}_0=0$ renders
$\breve{\bm \phi}^{(p)}=0$. Thus, one has that $\breve{\bm \upsilon}^{(p)}=0$,
$\breve{\bm\phi}^{(p)}=0$, $0\leq p\leq \pb$.  The argument can
be repeated for $\breve{\bm\upsilon}^{(\pb+1)}$ as
$\breve{\bm\upsilon}^{(\pb+1)}_0=0$.
\end{proof}

\bigskip
The ideas of the proof of Lemma
\ref{Lemma:LowerOrderSolutions}  will be used to study with some detail the
solutions to the sector $\mathfrak{S}_{\pb+1}$ transport equations for
the orders $\pb+1$, $\pb+2$, $\pb+3$ and $\pb+4$.

\subsection{The transport equations at order $p=\pb+1$}

From Lemma \ref{Lemma:LowerOrderSolutions} one knows that
${\bm\upsilon}^{(\pb+1)}=\mathring{\bm\upsilon}^{(\pb+1)}$. Thus, at this order
  one only needs to study the solutions to the Bianchi transport
  equations. 

\medskip
Substitution of the Ans\"atze 
\[
{\bm \upsilon}^{(j)}=\mathring{\bm \upsilon}^{(j)}+\breve{\bm
  \upsilon}^{(j)}, \quad {\bm \phi}^{(j)} = \mathring{\bm \phi}^{(j)}
+ \breve{\bm \phi}^{(j)}, \quad 0\leq j \leq \pb+1
\]
into equations \eqref{pbianchi:transport}-\eqref{cbianchi:transport}
with $p=\pb+1$ and using Lemma \ref{Lemma:LowerOrderSolutions} one
obtains the following equations for $\breve{\bm \phi}^{(\pb+1)}$:
 \begin{eqnarray*}
&& \sqrt{2}\mathbf{E} \cdot\partial_\tau \breve{\bm\phi}^{(p)} + (\mathring{\mathbf{A}}^\mu)^{(0)} \cdot\partial_\mu
\breve{\bm\phi}^{(p)} = \mathring{\mathbf{F}}^{(0)}\cdot\breve{\bm\phi}^{(p)}, \\
&& (\mathring{\mathbf{B}}^\mu)^{(0)} \cdot\partial_\mu \breve{\bm\phi}^{(p)} =
\mathring{\mathbf{G}}^{(0)} \cdot \breve{\bm \phi}^{(p)}. 
\end{eqnarray*}
We now focus on the sector $\mathfrak{S}_{\pb+1}$ of these
equations. Recalling that $(\mathring{\mathbf{A}}^\mu)^{(0)}$, $(\mathring{\mathbf{B}}^\mu)^{(0)}$, 
$\mathring{\mathbf{F}}^{(0)}$, $\mathring{\mathbf{G}}^{(0)}$ only contain the sector $\mathfrak{S}_0$, one obtains
the matricial equations
\begin{subequations}
\begin{eqnarray}
&& \mathbf{A}\cdot \breve{\bm \phi}'_{\pb+1;2(\pb+1)} +
\mathbf{A}_{\pb+1;2(\pb+1)}\cdot\breve{\bm \phi}_{\pb+1;2(\pb+1)}=0, \label{BianchiEvolution1}\\
&& \mathbf{B} \cdot\breve{\bm \phi}'_{\pb+1;2(\pb+1)} +
\mathbf{B}_{\pb+1;2(\pb+1)} \cdot\breve{\bm
  \phi}_{\pb+1;2(\pb+1)}=0, \label{BianchiConstraint1}
\end{eqnarray}
\end{subequations}
where $\mathbf{A}$, $\mathbf{A}_{\pb+1;2(\pb+1)}$, $\mathbf{B}$ and
$\mathbf{B}_{\pb+1;2(\pb+1)}$ are the matrices given by
\eqref{matrixA}-\eqref{matrixB}. For the sake of the simplicity of the
presentation, the subindex $k$ has been omitted from these and most of
the subsequent equations. A lengthy but straightforward computation shows that the initial data
for these equations is given by 
\begin{equation}
\breve{\bm{\phi}}_{\pb+1;2(\pb+1),k}(0) =(\breve{a}_{0,k},
\breve{a}_{1,k}, \breve{a}_{2,k}, \breve{a}_{1,k}, \breve{a}_{0,k}), \label{Data:p+1}
\end{equation}
with $k=0,\ldots,2(\pb+1)$ and
\begin{eqnarray*}
&& \breve{a}_{0,k} \equiv -\sqrt{\pb(\pb+1)(\pb+2)(\pb+3)}\breve{w}_{\pb+1;2(\pb+1),k}, \\
&& \breve{a}_{1,k} \equiv -4 (\pb+3)\sqrt{(\pb+1)(\pb+2)}\breve{w}_{\pb+1;2(\pb+1),k}, \\
&& \breve{a}_{2,k} \equiv -6 (\pb+2)(\pb+3) \breve{w}_{\pb+1;2(\pb+1),k}.
\end{eqnarray*}

The previous equations lead to the following crucial observation:

\begin{observation}
Equations \eqref{BianchiEvolution1}-\eqref{BianchiConstraint1} and
their corresponding initial data are formally identical to the sector
$\mathfrak{S}_{\pb+1}[{\bm\phi}^{(\pb+1)}]$ transport equations for data
which is Schwarzschildean up to order $\pb+1$; the
solutions are, therefore, also formally identical to those obtained in
\cite{Val10}.
\end{observation}

As a consequence of the latter observation one obtains from the analysis in
\cite{Val10} that: 
\[
\breve{\bm \phi}_{\pb+1;2(\pb+1)} = \tilde{\bm \varphi}_{\pb+1}(\tau) (1-\tau)^{\pb-1} (1+\tau)^{\pb-1},
\]
with $\tilde{\varphi}_{\pb+1}(\tau)$ having entries which are
polynomials of degree 4 in $\tau$. In particular, 
\[
\tilde{\bm \varphi}_{\pb+1}(0)= \breve{\bm{\phi}}_{\pb+1;2(\pb+1)}(0),
\]
as given by \eqref{Data:p+1}.

% \begin{subequations}
% \begin{eqnarray}
% && \hspace{-1cm} \breve{\phi}_{0,\pb+1;2(\pb+1),k}= \breve{a}_{0,k}(1-\tau)^{\pb+3}(1+\tau)^{\pb-1}, \label{Weyl:pp1_0}\\
% && \hspace{-1cm} \breve{\phi}_{1,\pb+1;2(\pb+1),k}= \breve{a}_{1,k} (1-\tau)^{\pb+2}(1+\tau)^{\pb}, \label{Weyl:pp1_1} \\
% && \hspace{-1cm} \breve{\phi}_{2,\pb+1;2(\pb+1),k}= \breve{a}_{2,k} (1-\tau)^{\pb+1}(1+\tau)^{\pb+1}, \\
% && \hspace{-1cm} \breve{\phi}_{3,\pb+1;2(\pb+1),k}= \breve{a}_{1,k} (1-\tau)^{\pb} (1+\tau)^{\pb+2}, \label{Weyl:pp1_3}\\
% && \hspace{-1cm} \breve{\phi}_{4,\pb+1;2(\pb+1),k}= \breve{a}_{0,k} (1-\tau)^{\pb-1}(1+\tau)^{\pb+3}, \label{Weyl:pp1_4}
% \end{eqnarray}
% \end{subequations}
%with $k=0,\ldots,2(\pb+1)$.

\medskip
Combining this analysis with Corollary
\ref{StaticCorollary} one finds the following result:

\begin{proposition}
\label{Lemma:pplus1}
The solution to the transport equations at spatial infinity at order
$\pb+1$ for data which is static up to order $\pb$ are polynomial in
$\tau$. Hence, they extend analytically through $\tau=\pm 1$. 
\end{proposition}

\subsection{The transport equations at order $p=\pb+2$}

Using Lemma \ref{Lemma:LogsLeadToLogs} one finds that the solutions
${\bm \upsilon}^{(\pb+2)}$ to the order $\pb+2$ ${\bm
  \upsilon}$-transport equations are polynomial in $\tau$ given that ${\bm \upsilon}^{(p)}$
and ${\bm \phi}^{(p)}$ for $0\leq p\leq \pb+1$ are polynomial in
$\tau$. However, we require more precise information. 

\medskip
Again, we consider the transport equations \eqref{upsilon:transport}
for the order $\pb+2$. The substitution of the Ans\"atze
\[
 {\bm\upsilon}^{(p)} = \mathring{\bm\upsilon}^{(p)} +
\breve{\bm\upsilon}^{(p)}, \quad
 {\bm\phi}^{(p)} = \mathring{\bm\phi}^{(p)} + \breve{\bm\phi}^{(p)},
\quad 0\leq p \leq \pb+2,
\]
and considerations similar to the ones used for the order $\pb+1$
lead to the following equations for $\breve{\bm\upsilon}^{(\pb+2)}$ and $\breve{\bm\phi}^{(\pb+2)}$:
\begin{eqnarray*}
&& \partial_\tau \breve{\bm\upsilon}^{(\pb+2)} = \mathbf{K}\cdot \breve{\bm\upsilon}^{(\pb+2)}
+ \mathbf{Q}(\mathring{\bm\upsilon}^{(0)}, \breve{\bm\upsilon}^{(\pb+2)}) +
\mathbf{Q}(\breve{\bm\upsilon}^{(\pb+2)},\mathring{\bm\upsilon}^{(0)}) +
\mathring{\mathbf{L}}^{(1)} \cdot\breve{\bm\phi}^{(\pb+1)}, \\
&& \sqrt{2}\mathbf{E}\cdot \partial_\tau \breve{\bm\phi}^{(\pb+2)} +
(\mathring{\mathbf{A}}^\mu)^{(0)}\cdot \partial_\mu \breve{\bm\phi}^{(\pb+2)} = \mathring{\mathbf{F}}^{(0)}
\cdot\breve{\bm\phi}_{\pb+2} + (\pb+2) \mathring{\mathbf{F}}^{(1)}\cdot
\breve{\bm\phi}^{(\pb+1)}  \\
&& \hspace{6cm}+ \breve{\mathbf{F}}^{(\pb+2)} \cdot
\mathring{\bm\phi}^{(0)} - (\pb+2) (\mathring{\mathbf{A}}^\mu)^{(1)}\cdot \partial_\mu
\breve{\bm\phi}^{(\pb+1)}, \\
&& (\mathring{\mathbf{B}}^\mu)^{(0)} \partial_\mu \breve{\bm\phi}^{(\pb+2)} = \mathring{\mathbf{G}}^{(0)}
\cdot\breve{\bm\phi}_{\pb+2} + (\pb+2) \mathring{\mathbf{G}}^{(1)}
\cdot\breve{\bm\phi}^{(\pb+1)}  \\
&& \hspace{6cm}+ \breve{\mathbf{G}}^{(\pb+2)} 
\cdot\mathring{\bm\phi}^{(0)} - (\pb+2) (\mathring{\mathbf{B}}^\mu)^{(1)}\cdot \partial_\mu
\breve{\bm\phi}^{(\pb+1)}. \\
\end{eqnarray*}
Again, as a consequence of Lemma \ref{Lemma:MultipoleStructureStatic}
one has that
\[
\mathring{\bm\upsilon}^{(0)}, \quad \mathring{\bm\phi}^{(0)}, \quad \mathring{\mathbf{L}}^{(1)}
\]
contain only contributions to the sector
$\mathfrak{S}_{0}$ so that one obtains directly the following equations for the
components of $\mathfrak{S}_{\pb+1}[\breve{\bm\upsilon}^{(\pb+2)},\breve{\bm\phi}^{(\pb+2)}]$:
\begin{subequations}
\begin{eqnarray}
&& \breve{\bm{\upsilon}}'_{\pb+2;2(\pb+1)} = \mathbf{K} \breve{\bm{\upsilon}}_{\pb+2;2(\pb+1)} +
\mathbf{Q}(\mathring{\bm{\upsilon}}_{0;0},
\breve{\bm{\upsilon}}_{\pb+2;2(\pb+1)}) \nonumber \\
&& \hspace{7mm} +
\mathbf{Q}(\breve{\bm{\upsilon}}_{\pb+2;2(\pb+1)},\mathring{\bm{\upsilon}}_{0;0}) +
\mathring{\mathbf{L}}_{1;0} \cdot\breve{\bm{\phi}}_{\pb+1;2(\pb+1)}, \label{Order2a}\\
&& (\mathbf{E}+\mathbf{A})\cdot \breve{\bm \phi}'_{\pb+2;2(\pb+1)} +
\mathbf{A}_{\pb+2;2(\pb+1)}\cdot \breve{\bm \phi}_{\pb+2;2(\pb+1)} =
\mathring{\mathbf{F}}_{0;0}\cdot 
\breve{\bm \phi}_{\pb+2;2(\pb+1)} \nonumber \\ 
&& \hspace{7mm}+ (\pb+2) \mathring{\mathbf{F}}_{1;0}
\cdot\breve{\bm \phi}_{\pb+1;2(\pb+1)}  + \breve{\mathbf{F}}_{\pb+2;2(\pb+1)} \cdot
\mathring{\bm \phi}_{0;0} - (\pb+2) \mathring{\mathbf{A}}_{1;0} \cdot
\breve{\bm \phi}_{\pb+1;2(\pb+1)},  \label{Order2b}\\
&&\mathbf{B} \cdot\breve{\bm \phi}'_{\pb+2;2(\pb+1)} +
\mathbf{B}_{\pb+2;2(\pb+1)}\cdot 
\breve{\bm \phi}_{\pb+2;2(\pb+1)} =  \mathring{\mathbf{G}}_{0;0}\cdot
\breve{\bm\phi}_{\pb+2} + (\pb+2) \mathring{\mathbf{G}}_{1;0}\cdot
\breve{\bm\phi}_{\pb+1;2(\pb+1)}  \nonumber \\
&& \hspace{7mm} + \breve{\mathbf{G}}_{\pb+2;2(\pb+1)} \cdot 
\mathring{\bm \phi}_{0;0} - (\pb+2) \mathring{\mathbf{B}}_{1;0} \cdot
\breve{\bm \phi}_{\pb+1;2(\pb+1)}. \label{Order2c}
\end{eqnarray}
\end{subequations}
Using Lemma \ref{Lemma:LowerOrderSolutions}, one sees that the components in
\[
\mathring{\bm\upsilon}_{0;0}, \; \mathring{\mathbf{L}}_{1;0}, \; 
\mathring{\mathbf{F}}_{0;0}, \; \mathring{\mathbf{G}}_{0;0},
\]
are exactly Minkowskian  while the components in
\[
\mathring{\bm\phi}_{0;0}, \; \mathring{\mathbf{A}}_{1;0}, \;
\mathring{\mathbf{B}}_{1;0}, \; \mathring{\mathbf{F}}_{1;0}, \; \mathring{\mathbf{G}}_{1;0} ,  
\]
are exactly Schwarzschildean ---that is, they depend only on the mass
$m$. On the other hand, as already seen, the components in
$\breve{\bm\phi}_{\pb+1;2(\pb+1)}$ depend only on
$\breve{w}_{\pb+1;2(\pb+1),k}$. This leads to the following crucial
observation at this order:

\begin{observation}
 The equations \eqref{Order2a}-\eqref{Order2c} are formally identical
to the order $\pb+2$ transport equations for initial data sets which
are Schwarzschildean up to order $\pb$; similarly, due to Lemma
\ref{LemmaQuasiStaticData} the initial data set is formally also of
the form of a perturbation of Schwarzschild.
\end{observation}

\smallskip
As a consequence of the previous discussion one can directly use the analysis and results of \cite{Val10} to
directly conclude that the solutions
$\breve{\bm\upsilon}_{\pb+2;2(\pb+1)},$ and
$\breve{\bm\phi}_{\pb+;2(\pb+1)},$  to equations
\eqref{Order2a}-\eqref{Order2c} are polynomial in $\tau$. Combining this observation with Corollary
\ref{StaticCorollary} one has that:

\begin{proposition}
The solutions to the order $\pb+2$ transport equations
\eqref{upsilon:transport}-\eqref{cbianchi:transport} are
polynomial in $\tau$ and, thus, extend analytically through
$\tau=\pm1$. 
\end{proposition}

\subsection{The transport equations at order $p=\pb+3$}
We now adapt the procedure discussed in the previous section to the
analysis of the order $\pb+3$ transport equations. As in the previous
order the polynomial dependence in $\tau$ follows directly from Lemma
\ref{Lemma:LogsLeadToLogs} once one knows that the entries in
${\bm\phi}^{(\pb+2)}$ are polynomial. Further detailed information will
follow from the analysis of the sector $\mathfrak{S}_{\pb+1}$.

\medskip
The analysis for this order is similar to that for orders $\pb+1$ and
$\pb+2$.  Substitution of the Ans\"atze 
\[
{\bm \upsilon}^{(p)} = \mathring{\bm \upsilon}^{(p)} + \breve{\bm \upsilon}^{(p)},
\quad {\bm \phi}^{(p)} = \mathring{\bm \phi}^{(p)} + \breve{\bm \phi}^{(p)},
\quad 0\leq p \leq \pb+3
\]
into the transport equations
\eqref{upsilon:transport}-\eqref{cbianchi:transport} leads to the
following equations for $\breve{\bm \upsilon}^{(\pb+3)}$ and
$\breve{\bm \phi}^{(\pb+3)}$:
\begin{eqnarray*}
&& \partial_\tau \breve{\bm \upsilon}^{(\pb+3)} = \mathbf{K}\cdot 
\breve{\bm \upsilon}^{(\pb+3)}
+ \mathbf{Q}(\mathring{\bm \upsilon}^{(0)}, \breve{\bm \upsilon}^{(\pb+3)}) +
\mathbf{Q}(\breve{\bm \upsilon}^{(\pb+3)},\mathring{\bm \upsilon}^{(0)})  \\
&& \hspace{3cm} +(\pb+3)  \mathbf{Q}(\mathring{\bm\upsilon}^{(1)}, \breve{\bm\upsilon}^{(\pb+2)}) + (\pb+3) \mathbf{Q}(\breve{\bm\upsilon}^{(\pb+2)},\mathring{\bm\upsilon}^{(1)})
+(\pb+3)\mathring{\mathbf{L}}^{(1)}\cdot \breve{\bm\phi}^{(\pb+2)}\\
&& \hspace{3cm}
+\tfrac{1}{2}(\pb+3)(\pb+2)\mathring{\mathbf{L}}^{(2)}\cdot\breve{\bm\phi}^{(\pb+1)}
+ \breve{\mathbf{L}}^{(\pb+3)}\cdot\mathring{\phi}^{(0)}, 
\end{eqnarray*}
\begin{eqnarray*}
&& \sqrt{2}\mathbf{E} \partial_\tau \breve{\bm\phi}^{(\pb+3)} +
(\mathring{\mathbf{A}}^\mu)^{(0)} \cdot\partial_\mu \breve{\bm\phi}^{(\pb+3)} = \mathring{\mathbf{F}}^{(0)}\cdot
\breve{\bm\phi}_{\pb+3} + (\pb+3) \mathring{\mathbf{F}}^{(1)}\cdot
\breve{\bm\phi}^{(\pb+2)}  \\
&& \hspace{3cm}+\frac{1}{2}(\pb+3)(\pb+2) \mathring{\mathbf{F}}^{(2)}\cdot
\breve{\bm\phi}^{(\pb+1)}
+(\pb+3) \breve{\mathbf{F}}^{(\pb+2)} \cdot\mathring{\bm\phi}^{(1)}+
\breve{\mathbf{F}}^{(\pb+3)} \cdot\mathring{\bm\phi}^{(0)} \\
&& \hspace{3cm}- (\pb+3) (\mathring{\mathbf{A}}^\mu)^{(1)} \cdot\partial_\mu
\breve{\bm\phi}^{(\pb+2)}-\frac{1}{2}(\pb+3)(\pb+2) (\mathring{\mathbf{A}}^\mu)^{(2)}\cdot \partial_\mu
\breve{\bm\phi}^{(\pb+1)}\\
&& \hspace{3cm}-(\pb+3) (\breve{\mathbf{A}}^\mu)^{(\pb+2)}\cdot \partial_\mu
\mathring{\bm\phi}^{(1)}-(\breve{\mathbf{A}}^\mu)^{(\pb+3)} \cdot\partial_\mu
\mathring{\bm\phi}^{(0)}, 
\end{eqnarray*}
\begin{eqnarray*}
&& (\mathbf{B}^\mu)^{(0)} \partial_\mu \breve{\bm\phi}^{(\pb+3)} = \mathring{\mathbf{G}}^{(0)}\cdot
\breve{\bm\phi}_{\pb+3} + (\pb+3) \mathring{\mathbf{G}}^{(1)}\cdot
\breve{\bm\phi}^{(\pb+2)} +\frac{1}{2}(\pb+3)(\pb+2)\mathring{\mathbf{G}}^{(2)}\cdot \breve{\bm\phi}^{(\pb+1)}\\
&& \hspace{3cm}+ \breve{\mathbf{G}}^{(\pb+3)}\cdot 
\mathring{\bm\phi}^{(0)} +(\pb+3) \breve{\mathbf{G}}^{(\pb+2)}\cdot 
\mathring{\bm\phi}^{(1)}\\
 && \hspace{3cm}- (\pb+3) (\mathring{\mathbf{B}}^\mu)^{(1)} \cdot\partial_\mu
\breve{\bm\phi}^{(\pb+2)} -\frac{1}{2}(\pb+3)(\pb+2) (\mathring{\mathbf{B}}^\mu)^{(2)}\cdot \partial_\mu
\breve{\bm\phi}^{(\pb+1)} \\
&& \hspace{3cm} -(\pb+3) (\breve{\mathbf{B}}^\mu)^{(\pb+2)} \cdot\partial_\mu
\mathring{\bm\phi}^{(1)}-(\breve{\mathbf{B}}^\mu)^{(\pb+3)}\cdot \partial_\mu
\mathring{\bm\phi}^{(0)},
\end{eqnarray*}
Using Lemma
\ref{Lemma:MultipoleStructureStatic} one obtains the following
equations for the components of
$\mathfrak{S}_{\pb+1}[\breve{\bm\upsilon}^{(\pb+3)},\breve{\bm \phi}^{(\pb+3)}]$:
\begin{subequations}
\begin{eqnarray}
&& \breve{\bm \upsilon}'_{\pb+3;2(\pb+1)} = \mathbf{K}\cdot 
\breve{\bm \upsilon}_{\pb+3;2(\pb+1)}
+ \mathbf{Q}(\mathring{\bm \upsilon}_{0;0}, \breve{\bm \upsilon}_{\pb+3;2(\pb+1)}) +
\mathbf{Q}(\breve{\bm \upsilon}_{\pb+3;2(\pb+1)},\mathring{\bm
  \upsilon}_{0;0}) \nonumber \\
&& \hspace{3cm} +(\pb+3)  \mathbf{Q}(\mathring{\bm \upsilon}_{1;0},
\breve{\bm\upsilon}_{\pb+2;2(\pb+1)}) + (\pb+3) \mathbf{Q}(\breve{\bm
  \upsilon}_{\pb+2;2(\pb+1)},\mathring{\bm \upsilon}_{1;0}) \nonumber\\
&& \hspace{3cm}+(\pb+3)\mathring{\mathbf{L}}_{1;0}\cdot \breve{\bm \phi}_{\pb+2;2(\pb+1)}
+\tfrac{1}{2}(\pb+3)(\pb+2)\mathring{\mathbf{L}}_{2;0}\cdot\breve{\bm\phi}_{\pb+1;2(\pb+1)}\nonumber
\\
&& \hspace{3cm}+
\mathring{\mathbf{L}}_{\pb+3;2(\pb+1)}\cdot\mathring{\bm\phi}_{0;0}, \label{Order3a}
\end{eqnarray}
\begin{eqnarray}
&& (\mathbf{E}+\mathbf{A})\cdot\breve{\bm\phi}'_{\pb+3;2(\pb+1)} +
\mathbf{A}_{\pb+3;2(\pb+1)} \cdot \breve{\bm \phi}_{\pb+3;2(\pb+1)} = \mathring{\mathbf{F}}_{0;0}
\cdot \breve{\bm\phi}_{\pb+3;2(\pb+1)} + (\pb+3) \mathring{\mathbf{F}}_{1;0}\cdot
\breve{\bm\phi}_{\pb+2;2(\pb+1)} \nonumber \\
&& \hspace{3cm}+\tfrac{1}{2}(\pb+3)(\pb+2) \mathring{\mathbf{F}}_{2;0}
\cdot\breve{\bm\phi}_{\pb+1;2(\pb+1)}
+(\pb+3) \breve{\mathbf{F}}_{\pb+2;2(\pb+1)} \cdot
\mathring{\bm\phi}_{1;0} \nonumber\\
&& \hspace{3cm}+  \breve{\mathbf{F}}_{\pb+3;2(\pb+1)}
\cdot\mathring{\bm \phi}_{0;0}  \nonumber\\
&& \hspace{3cm}- (\pb+3) \mathring{\mathbf{A}}_{1;0} \cdot
\breve{\bm \phi}'_{\pb+2;2(\pb+1)}
- (\pb+3) \mathring{\mathbf{A}}^+_{1;0} \cdot
\breve{\bm \phi}_{\pb+2;2(\pb+1)}
- (\pb+3) \mathring{\mathbf{A}}^-_{1;0} \cdot
\breve{\bm \phi}_{\pb+2;2(\pb+1)} \nonumber\\
&& \hspace{3cm}
-\tfrac{1}{2}(\pb+3)(\pb+2) \mathring{\mathbf{A}}_{2;0} \cdot
\breve{\bm \phi}'_{\pb+1;2(\pb+1)}
-\tfrac{1}{2}(\pb+3)(\pb+2) \mathring{\mathbf{A}}^+_{2;0} \cdot
\breve{\bm \phi}_{\pb+1;2(\pb+1)} \nonumber\\
&& \hspace{3cm}
-\tfrac{1}{2}(\pb+3)(\pb+2) \mathring{\mathbf{A}}^-_{2;0} \cdot
\breve{\bm \phi}_{\pb+1;2(\pb+1)} -(\pb+3) \breve{\mathbf{A}}_{\pb+2;2(\pb+1)} 
\cdot \mathring{\bm \phi}'_{1;0}, \label{Order3b}
\end{eqnarray}
\begin{eqnarray}
&& \mathbf{B}\cdot \breve{\bm
  \phi}'_{\pb+3;2(\pb+1)}+\mathbf{B}_{\pb+3;2(\pb+1)}\cdot \breve{\bm\phi}_{\pb+3;2(\pb+1)} = \mathring{\mathbf{G}}_{0;0}
\cdot \breve{\bm\phi}_{\pb+3;2(\pb+1)} + (\pb+3) \mathring{\mathbf{G}}_{1;0}\cdot
\breve{\bm\phi}_{\pb+2;2(\pb+1)} \nonumber\\
&& \hspace{3cm}+\tfrac{1}{2}(\pb+3)(\pb+2)\mathring{\mathbf{G}}_{2;0} \breve{\bm\phi}_{\pb+1;2(\pb+1)}+ \breve{\mathbf{G}}_{\pb+3;2(\pb+1)}\cdot 
\mathring{\bm \phi}_{0;0} +(\pb+3) \breve{\mathbf{G}}_{\pb+2;2(\pb+1)} \cdot
\mathring{\bm\phi}_{1;0}\nonumber\\
 && \hspace{3cm}
- (\pb+3) \mathring{\mathbf{B}}_{1;0}\cdot \breve{\bm\phi}'_{\pb+2;2(\pb+1)}
- (\pb+3) \mathring{\mathbf{B}}_{1;0}^+\cdot \breve{\bm\phi}_{\pb+2;2(\pb+1)}
- (\pb+3) \mathring{\mathbf{B}}_{1;0}^-\cdot \breve{\bm\phi}_{\pb+2;2(\pb+1)}
\nonumber\\
&& \hspace{3cm}
-\tfrac{1}{2}(\pb+3)(\pb+2) \mathring{\mathbf{B}}_{2;0} \cdot\breve{\bm
  \phi}'_{\pb+1;2(\pb+1)}
-\tfrac{1}{2}(\pb+3)(\pb+2) \mathring{\mathbf{B}}^+_{2;0} \cdot\breve{\bm
  \phi}_{\pb+1;2(\pb+1)} \nonumber\\
&&\hspace{3cm}
-\tfrac{1}{2}(\pb+3)(\pb+2) \mathring{\mathbf{B}}^-_{2;0} \cdot\breve{\bm
  \phi}_{\pb+1;2(\pb+1)}  -(\pb+3)\breve{\mathbf{B}}_{\pb+2;2(\pb+1)} \cdot
\mathring{\bm \phi}'_{1;0}. \label{Order3c}
\end{eqnarray}
\end{subequations}
As in the analysis of lower order transport equations, we begin
by noticing that the entries of
\[
\mathring{\bm \upsilon}_{0;0}, \quad \mathring{\mathbf{L}}_{1;0}, \quad
\mathring{\mathbf{F}}_{0;0}, \quad \mathring{\mathbf{G}}_{0;0}
\]
are Minkowskian ---i.e. independent of $m$ and $W$. On the other hand,
the entries of
\begin{eqnarray*}
&& \mathring{\bm \upsilon}_{1;0}, \quad \mathring{\mathbf{L}}_{2;0}, \quad
\mathring{\bm\phi}_{0;0}, \quad \mathring{\bm\phi}_{1;0}, \\
&& \mathring{\mathbf{F}}_{1;0}, \quad \mathring{\mathbf{F}}_{2;0}, \quad
\mathring{\mathbf{A}}_{1;0}, \quad \mathring{\mathbf{A}}^\pm_{1;0}, \quad
\mathring{\mathbf{A}}_{2;0}, \quad \mathring{\mathbf{A}}^\pm_{2;0}, \\
&& \mathring{\mathbf{G}}_{1;0}, \quad \mathring{\mathbf{G}}_{2;0}, \quad
\mathring{\mathbf{B}}_{1;0}, \quad \mathring{\mathbf{B}}^\pm_{1;0}, \quad
\mathring{\mathbf{B}}_{2;0}, \quad \mathring{\mathbf{B}}^\pm_{2;0},
\end{eqnarray*}
are Schwarzschildean in our gauge---i.e. they depend only on $m$.

\medskip
As a consequence of the analysis of the transport equations
for the orders $\pb+1$ and $\pb+2$ one has that the entries in
\[
\breve{\bm \upsilon}_{\pb+1;2(\pb+1)}, \quad \breve{\bm \phi}_{\pb+1;2(\pb+1)}, \quad \breve{\bm \phi}_{\pb+2;2(\pb+1)},
\]
depend only on $m$ and the coefficients
$\breve{w}_{\pb+1;2(\pb+1);k}$. As in the analysis of the orders
$\pb+1$ and $\pb+2$ one has the crucial observation:

\begin{observation}
The equations \eqref{Order3a}-\eqref{Order3c} and their initial
  conditions are formally
  identical to the order $\pb+3$ transport equations for initial data
  sets which are Schwarzschildean up to order $\pb$.
\end{observation}

As in the analysis of the orders $\pb+1$, $\pb+2$ and $\pb+3$, one 
can use directly the analysis in \cite{Val10} to
conclude that the solutions $\mathring{\bm\upsilon}_{\pb+3;2(\pb+1)}$ and
$\mathring{\bm\phi}_{\pb+3;2(\pb+1)}$ to equations
\eqref{Order3a}-\eqref{Order3c} are polynomial in $\tau$. More
generally, one has that:

\begin{proposition}
The solutions to the order $\pb+3$ transport equations
\eqref{upsilon:transport}-\eqref{cbianchi:transport} are polynomial in
$\tau$ and, thus, extend analytically through $\tau=\pm 1$. 
\end{proposition}

\subsection{The transport equations at order $p=\pb+4$}

Finally, using similar methods, we discuss the
order $\pb+4$ transport equations. As in the previous
order the polynomial dependence in $\tau$ follows directly from Lemma
\ref{Lemma:LogsLeadToLogs} once one knows that the entries in
${\bm\phi}^{(\pb+3)}$ are polynomial. The rest of the analysis is more
subtle than at lower orders. The reason for this is twofold:

\begin{itemize}

\item[(i)] As a consequence of Lemma \ref{LemmaQuasiStaticData} one
  finds that $\breve{\phi}_{ABCD}^{(\pb+3)}$ has a multipolar
  structure which is more complicated than that of data which just
  deviates from Schwarzschild data.

\item[(ii)] The discussion of the transport equations for the sectors
  in $\mathfrak{S}_{\pb+1}$ involves non-trivial multiplications of
  the functions $\TT{i}{j}{k}$. 
\end{itemize}
Evidence from explicit calculations suggests that the solutions to the
Bianchi transport equations at this order are, generically, not
smooth at the critical sets. 

\medskip
Proceeding as in the case of lower order one finds that the substitution of the Ans\"atze 
\[
{\bm \upsilon}^{(p)} = \mathring{\bm \upsilon}^{(p)} + \breve{\bm \upsilon}^{(p)},
\quad {\bm \phi}^{(p)} = \mathring{\bm \phi}^{(p)} + \breve{\bm \phi}^{(p)},
\quad 0\leq p \leq \pb+4
\]
into the transport equations
\eqref{upsilon:transport}-\eqref{cbianchi:transport} and taking into
account Lemma \ref{Lemma:HigherOrderL}  leads to the
following equations for $\breve{\bm \upsilon}^{(\pb+4)}$ and
$\breve{\bm \phi}^{(\pb+4)}$:
\begin{eqnarray*}
&& \partial_\tau \breve{\bm \upsilon}^{(\pb+4)} = \mathbf{K}\cdot 
\breve{\bm \upsilon}^{(\pb+4)}
+ \mathbf{Q}(\mathring{\bm \upsilon}^{(0)}, \breve{\bm \upsilon}^{(\pb+4)}) +
\mathbf{Q}(\breve{\bm \upsilon}^{(\pb+4)},\mathring{\bm \upsilon}^{(0)})  \\
&& \hspace{3cm} +(\pb+4) \left(  \mathbf{Q}(\mathring{\bm\upsilon}^{(1)},
\breve{\bm\upsilon}^{(\pb+3)}) +
\mathbf{Q}(\breve{\bm\upsilon}^{(\pb+3)},\mathring{\bm\upsilon}^{(1)}) \right)
\\
&& \hspace{3cm} +\tfrac{1}{2}(\pb+4)(\pb+3) \left(\mathbf{Q}(\mathring{\bm\upsilon}^{(2)},
\breve{\bm\upsilon}^{(\pb+2)}) + \mathbf{Q}(\breve{\bm\upsilon}^{(\pb+2)},\mathring{\bm\upsilon}^{(2)}) \right)\\
&& \hspace{3cm} + \breve{\mathbf{L}}^{(\pb+4)}\cdot\mathring{\bm\phi}^{(0)} + (\pb+4)
\breve{\mathbf{L}}^{(\pb+3)}\cdot\mathring{\bm\phi}^{(1)}\\
&& \hspace{3cm}+(\pb+4)\mathring{\mathbf{L}}^{(1)}\cdot \breve{\bm
  \phi}^{(\pb+3)}
+\tfrac{1}{2}(\pb+4)(\pb+3)\mathring{\mathbf{L}}^{(2)}\cdot\breve{\bm\phi}^{(\pb+2)}
\\
&& \hspace{3cm}+
\tfrac{1}{6}(\pb+4)(\pb+3)(\pb+2)\mathring{\mathbf{L}}^{(3)}\cdot\mathring{\phi}^{(\pb+1)},
\end{eqnarray*}
\begin{eqnarray*}
&& \sqrt{2}\mathbf{E} \partial_\tau \breve{\phi}^{(\pb+4)} +
(\mathring{\mathbf{A}}^\mu)^{(0)} \partial_\mu \breve{\bm\phi}^{(\pb+4)} = \mathring{\mathbf{F}}^{(0)}\cdot
\breve{\bm\phi}_{\pb+4} + (\pb+4) \mathring{\mathbf{F}}^{(1)}\cdot
\breve{\bm\phi}^{(\pb+3)}  \\
&& \hspace{3cm}+\tfrac{1}{2}(\pb+4)(\pb+4) \mathring{\mathbf{F}}^{(2)}\cdot
\breve{\bm\phi}^{(\pb+2)} + \tfrac{1}{6}(\pb+4)(\pb+3)(\pb+2) \mathring{\mathbf{F}}^{(3)}\cdot\breve{\bm\phi}^{(\pb+1)}\\
&& \hspace{3cm}
+\tfrac{1}{2}(\pb+4)(\pb+3)\breve{\mathbf{F}}^{(\pb+2)}\cdot\mathring{\bm\phi}^{(2)}+(\pb+4) \breve{\mathbf{F}}^{(\pb+3)} \cdot\mathring{\bm\phi}^{(1)}+
\breve{\mathbf{F}}^{(\pb+4)}\cdot \mathring{\bm\phi}^{(0)} \\
&& \hspace{3cm}- (\pb+4) (\mathring{\mathbf{A}}^\mu)^{(1)} \partial_\mu
\breve{\bm\phi}^{(\pb+3)}-\tfrac{1}{2}(\pb+4)(\pb+3) (\mathring{\mathbf{A}}^\mu)^{(2)} \partial_\mu
\breve{\bm\phi}^{(\pb+2)}\\
&& \hspace{3cm}-\tfrac{1}{6}(\pb+4)(\pb+3)(\pb+2)  (\mathring{\mathbf{A}}^\mu)^{(3)} \partial_\mu
\breve{\bm\phi}^{(\pb+1)}\\
&& \hspace{3cm}-\tfrac{1}{2}(\pb+4)(\pb+3) (\breve{\mathbf{A}}^\mu)^{(\pb+2)} \partial_\mu
\mathring{\bm\phi}^{(2)}-(\pb+4) (\breve{\mathbf{A}}^\mu)^{(\pb+3)} \partial_\mu
\mathring{\bm\phi}^{(1)}-(\breve{\mathbf{A}}^\mu)^{(\pb+4)} \partial_\mu
\mathring{\bm\phi}^{(0)}, 
\end{eqnarray*}
\begin{eqnarray*}
&& (\mathring{\mathbf{B}}^\mu)^{(0)} \partial_\mu \breve{\bm\phi}^{(\pb+4)} = \mathring{\mathbf{G}}^{(0)}\cdot
\breve{\bm\phi}_{\pb+4} + (\pb+4) \mathring{\mathbf{G}}^{(1)}\cdot
\breve{\bm\phi}^{(\pb+3)}  \\
&& \hspace{3cm}+\tfrac{1}{2}(\pb+4)(\pb+3) \mathring{\mathbf{G}}^{(2)}\cdot
\breve{\bm\phi}^{(\pb+2)} + \tfrac{1}{6}(\pb+4)(\pb+3)(\pb+2) \mathring{\mathbf{G}}^{(3)}\cdot\breve{\bm\phi}^{(\pb+1)}\\
&& \hspace{3cm}
+\tfrac{1}{2}(\pb+4)(\pb+3)\breve{\mathbf{G}}^{(\pb+2)}\cdot\mathring{\bm\phi}^{(2)}+(\pb+4) \breve{\mathbf{G}}^{(\pb+3)} \cdot\mathring{\bm\phi}^{(1)}+
\breve{\mathbf{G}}^{(\pb+4)}\cdot \mathring{\bm\phi}^{(0)} \\
&& \hspace{3cm}- (\pb+4) (\mathring{\mathbf{B}}^\mu)^{(1)} \partial_\mu
\breve{\bm\phi}^{(\pb+3)}-\tfrac{1}{2}(\pb+4)(\pb+3) (\mathring{\mathbf{B}}^\mu)^{(2)} \partial_\mu
\breve{\bm\phi}^{(\pb+2)}\\
&& \hspace{3cm}-\tfrac{1}{6}(\pb+4)(\pb+3)(\pb+2)  (\mathring{\mathbf{B}}^\mu)^{(3)} \partial_\mu
\breve{\bm\phi}^{(\pb+1)}\\
&& \hspace{3cm}-\tfrac{1}{2}(\pb+4)(\pb+3) (\breve{\mathbf{B}}^\mu)^{(\pb+2)} \partial_\mu
\mathring{\bm\phi}^{(2)}-(\pb+4) (\breve{\mathbf{B}}^\mu)^{(\pb+3)} \partial_\mu
\mathring{\bm\phi}^{(1)}-(\breve{\mathbf{B}}^\mu)^{(\pb+4)} \partial_\mu
\mathring{\bm\phi}^{(0)}.
\end{eqnarray*}

As in the previous orders, we extract 
equations for the components of
$\mathfrak{S}_{\pb+1}[{\bm\upsilon}^{(\pb+4)},{\bm\phi}^{(\pb+4)}]$.
In this case the analysis is more involved as there are terms involving
non-trivial products of spherical harmonics.  An inspection shows that
the terms containing this type of non-trivial products are
\begin{eqnarray*}
&& \mathring{\mathbf{F}}^{(3)}\cdot \breve{\bm \phi}^{(\pb+1)}, \quad
\breve{\mathbf{F}}^{(\pb+2)}\cdot \mathring{\bm \phi}^{(2)}, \quad
(\mathring{\mathbf{A}}^\mu)^{(3)}\cdot\partial_\mu \breve{\bm \phi}^{(\pb+1)} , \quad
(\breve{\mathbf{A}}^\mu)^{(\pb+2)}\cdot\partial_\mu\breve{\bm
  \phi}^{(2)}, \\
&& \mathring{\mathbf{G}}^{(3)}\cdot \breve{\bm \phi}^{(\pb+1)}, \quad
\breve{\mathbf{G}}^{(\pb+2)}\cdot \mathring{\bm \phi}^{(2)}, \quad
(\mathring{\mathbf{B}}^\mu)^{(3)}\cdot\partial_\mu \breve{\bm \phi}^{(\pb+1)} , \quad
(\breve{\mathbf{B}}^\mu)^{(\pb+2)}\cdot\partial_\mu\breve{\bm
  \phi}^{(2)}.
\end{eqnarray*}
All these terms contain products of the form
$\TT{4}{j_1}{k_1}\times \TT{2(\pb+1)}{i_2}{j_2}$, which using formula
\eqref{Clebsch-Gordan} can be linearised to render
\begin{eqnarray*}
&&\TT{4}{j_1}{k_1}\times \TT{2(\pb+1)}{i_2}{j_2}  \\
&&\hspace{1cm}=c_{2(\pb+3),j_1,j_2,k_1,k_2}\TT{2(\pb+3)}{j_1+j_2}{k_1+k_2}
+ c_{2(\pb+1),j_1,j_2,k_1,k_2}\TT{2(\pb+1)}{j_1+j_2-1}{k_1+k_2-1} \\
&&\hspace{1.5cm}+  c_{2(\pb-1),j_1,j_2,k_1,k_2}\TT{2(\pb-1)}{j_1+j_2-2}{k_1+k_2-2},
\end{eqnarray*}
with $c_{2(\pb+3),j_1,j_2,k_1,k_2}$, $c_{2(\pb+1),j_1,j_2,k_1,k_2}$,
$c_{2(\pb-1),j_1,j_2,k_1,k_2}$ some numerical coefficients. Their
explicit form will not be essential for the subsequent analysis. Thus,
the sector $\mathfrak{S}_{\pb+1}$ of, say,
$\mathring{\mathbf{F}}^{(3)}\cdot \breve{\bm \phi}^{(\pb+1)}$ is of
the form
\[
\mathring{\mathbf{F}}_{3;0}\cdot \breve{\bm \phi}_{\pb+1;2(\pb+1)} + \mathbf{C}_{4,2(\pb+1);\pb+1}\cdot\mathring{\mathbf{F}}_{3;4}\cdot \breve{\bm \phi}_{\pb+1;2(\pb+1)},
\]
where $\mathbf{C}_{4,2(\pb+1);\pb+1}$ denotes a matrix with numerical
entries. The term $\mathring{\mathbf{F}}_{3;0}\cdot \breve{\bm
  \phi}_{\pb+1;2(\pb+1)}$ is formally identical to the one one would
obtain from considering asymptotically Schwarzschildean data. A
similar analysis can be carried out with the other terms containing 
non-trivial products. 

\medskip
The equations for the components of
$\mathfrak{S}_{\pb+1}[\breve{\bm\upsilon}^{(\pb+4)},\breve{\bm\phi}^{(\pb+4)}]$
are:
\begin{subequations}
\begin{eqnarray}
&& \breve{\bm \upsilon}'_{\pb+4;2(\pb+1)} = \mathbf{K}\cdot 
\breve{\bm \upsilon}_{\pb+4;2(\pb+1)}
+ \mathbf{Q}(\mathring{\bm \upsilon}_{0;0}, \breve{\bm \upsilon}_{\pb+4;2(\pb+1)}) +
\mathbf{Q}(\breve{\bm \upsilon}_{\pb+4;2(\pb+1)},\mathring{\bm \upsilon}_{0;0})  \nonumber\\
&& \hspace{3cm} +(\pb+4) \left( \mathbf{Q}(\mathring{\bm \upsilon}_{1;0},
\breve{\bm\upsilon}_{\pb+3;2(\pb+1)}) +  \mathbf{Q}(\breve{\bm
  \upsilon}_{\pb+3;2(\pb+1)},\mathring{\bm \upsilon}_{1;0})\right) \nonumber\\
&& \hspace{3cm} +\tfrac{1}{2}(\pb+4)(\pb+3) \left( \mathbf{Q}(\mathring{\bm \upsilon}_{2;0},
\breve{\bm\upsilon}_{\pb+2;2(\pb+1)}) +  \mathbf{Q}(\breve{\bm
  \upsilon}_{\pb+2;2(\pb+1)},\mathring{\bm \upsilon}_{2;0})\right) \nonumber\\
&& \hspace{3cm}+(\pb+4)\mathring{\mathbf{L}}_{1;0}\cdot \breve{\bm \phi}_{\pb+3;2(\pb+1)}
+\tfrac{1}{2}(\pb+4)(\pb+3)\mathring{\mathbf{L}}_{2;0}\cdot\breve{\bm\phi}_{\pb+2;2(\pb+1)}\nonumber\\
&& \hspace{3cm} +\tfrac{1}{6}(\pb+4)(\pb+3)(\pb+2)\mathring{\mathbf{L}}_{3;0}\cdot\breve{\bm\phi}_{\pb+1;2(\pb+1)}
\nonumber\\
&& \hspace{3cm}+
\mathring{\mathbf{L}}_{\pb+4;2(\pb+1)}\cdot\mathring{\bm\phi}_{0;0}+
(\pb+4)
\mathring{\mathbf{L}}_{\pb+3;2(\pb+1)}\cdot\mathring{\bm\phi}_{1;0}, \label{Order4a}
\end{eqnarray}
\begin{eqnarray}
&& (\mathbf{E}+\mathbf{A})\cdot\breve{\bm\phi}'_{\pb+4;2(\pb+1)} +
\mathbf{A}_{\pb+4;2(\pb+1)} \cdot \breve{\bm \phi}_{\pb+4;2(\pb+1)} = \mathring{\mathbf{F}}_{0;0}
\cdot \breve{\bm\phi}_{\pb+4;2(\pb+1)} + (\pb+4) \mathring{\mathbf{F}}_{1;0}\cdot
\breve{\bm\phi}_{\pb+3;2(\pb+1)} \nonumber \\
&& \hspace{2cm}+\tfrac{1}{2}(\pb+4)(\pb+3) \mathring{\mathbf{F}}_{2;0}
\cdot\breve{\bm\phi}_{\pb+2;2(\pb+1)} +
\tfrac{1}{6}(\pb+4)(\pb+3)(\pb+2) \mathring{\mathbf{F}}_{3;0}\cdot \breve{\phi}_{\pb+1;2(\pb+1)} \nonumber\\
&& \hspace{2cm}+  \breve{\mathbf{F}}_{\pb+4;2(\pb+1)}
\cdot\mathring{\bm \phi}_{0;0} +(\pb+4) \breve{\mathbf{F}}_{\pb+3;2(\pb+1)} \cdot
\mathring{\bm\phi}_{1;0} +
\tfrac{1}{2}(\pb+4)(\pb+3)\breve{\mathbf{F}}_{\pb+2;2(\pb+1)}\cdot
\mathring{\bm \phi}_{2;0} \nonumber\\
&& \hspace{2cm}- (\pb+4) \left(\mathring{\mathbf{A}}_{1;0} \cdot
\breve{\bm \phi}'_{\pb+3;2(\pb+1)}
+\mathring{\mathbf{A}}^+_{1;0} \cdot
\breve{\bm \phi}_{\pb+3;2(\pb+1)}
+ \mathring{\mathbf{A}}^-_{1;0} \cdot
\breve{\bm \phi}_{\pb+3;2(\pb+1)} \right)\nonumber\\
&& \hspace{2cm}
-\tfrac{1}{2}(\pb+4)(\pb+3) \left(\mathring{\mathbf{A}}_{2;0} \cdot
\breve{\bm \phi}'_{\pb+2;2(\pb+1)}
+\mathring{\mathbf{A}}^+_{2;0} \cdot
\breve{\bm \phi}_{\pb+2;2(\pb+1)} 
+\mathring{\mathbf{A}}^-_{2;0} \cdot
\breve{\bm \phi}_{\pb+2;2(\pb+1)}\right) \nonumber \\
&& \hspace{2cm} -\tfrac{1}{6}(\pb+4)(\pb+3)(\pb+2) \left(  \mathring{\mathbf{A}}_{3;0} \cdot
\breve{\bm \phi}'_{\pb+1;2(\pb+1)}
+\mathring{\mathbf{A}}^+_{3;0} \cdot
\breve{\bm \phi}_{\pb+1;2(\pb+1)} 
+\mathring{\mathbf{A}}^-_{3;0} \cdot
\breve{\bm \phi}_{\pb+1;2(\pb+1)}\right) \nonumber \\
&& \hspace{2cm} -(\pb+4) \breve{\mathbf{A}}_{\pb+3;2(\pb+1)} 
\cdot \mathring{\bm \phi}'_{1;0}-\tfrac{1}{2}(\pb+4)(\pb+3) \breve{\mathbf{A}}_{\pb+2;2(\pb+1)} 
\cdot \mathring{\bm \phi}'_{2;0} \nonumber\\
&& \hspace{2cm} + \mathbf{R}_{2(\pb+1);\pb+1}, \label{Order4b}
\end{eqnarray}
\begin{eqnarray}
&& \mathbf{B}\cdot\breve{\bm\phi}'_{\pb+4;2(\pb+1)} +
\mathbf{B}_{\pb+4;2(\pb+1)} \cdot \breve{\bm \phi}_{\pb+4;2(\pb+1)} = \mathring{\mathbf{G}}_{0;0}
\cdot \breve{\bm\phi}_{\pb+4;2(\pb+1)} + (\pb+4) \mathring{\mathbf{G}}_{1;0}\cdot
\breve{\bm\phi}_{\pb+3;2(\pb+1)} \nonumber \\
&& \hspace{2cm}+\tfrac{1}{2}(\pb+4)(\pb+3) \mathring{\mathbf{G}}_{2;0}
\cdot\breve{\bm\phi}_{\pb+2;2(\pb+1)} +
\tfrac{1}{6}(\pb+4)(\pb+3)(\pb+2) \mathring{\mathbf{G}}_{3;0}\cdot \breve{\phi}_{\pb+1;2(\pb+1)} \nonumber\\
&& \hspace{2cm}+  \breve{\mathbf{G}}_{\pb+4;2(\pb+1)}
\cdot\mathring{\bm \phi}_{0;0} +(\pb+4) \breve{\mathbf{G}}_{\pb+3;2(\pb+1)} \cdot
\mathring{\bm\phi}_{1;0} +
\tfrac{1}{2}(\pb+4)(\pb+3)\breve{\mathbf{G}}_{\pb+2;2(\pb+1)}\cdot
\mathring{\bm \phi}_{2;0} \nonumber\\
&& \hspace{2cm}- (\pb+4) \left(\mathring{\mathbf{B}}_{1;0} \cdot
\breve{\bm \phi}'_{\pb+3;2(\pb+1)}
+\mathring{\mathbf{B}}^+_{1;0} \cdot
\breve{\bm \phi}_{\pb+3;2(\pb+1)}
+ \mathring{\mathbf{B}}^-_{1;0} \cdot
\breve{\bm \phi}_{\pb+3;2(\pb+1)} \right)\nonumber\\
&& \hspace{2cm}
-\tfrac{1}{2}(\pb+4)(\pb+3) \left(\mathring{\mathbf{B}}_{2;0} \cdot
\breve{\bm \phi}'_{\pb+2;2(\pb+1)}
+\mathring{\mathbf{B}}^+_{2;0} \cdot
\breve{\bm \phi}_{\pb+2;2(\pb+1)} 
+\mathring{\mathbf{B}}^-_{2;0} \cdot
\breve{\bm \phi}_{\pb+2;2(\pb+1)}\right) \nonumber \\
&& \hspace{2cm} -\tfrac{1}{6}(\pb+4)(\pb+3)(\pb+2) \left(  \mathring{\mathbf{B}}_{3;0} \cdot
\breve{\bm \phi}'_{\pb+1;2(\pb+1)}
+\mathring{\mathbf{B}}^+_{3;0} \cdot
\breve{\bm \phi}_{\pb+1;2(\pb+1)} 
+\mathring{\mathbf{B}}^-_{3;0} \cdot
\breve{\bm \phi}_{\pb+1;2(\pb+1)}\right) \nonumber \\
&& \hspace{2cm} -(\pb+4) \breve{\mathbf{B}}_{\pb+3;2(\pb+1)} 
\cdot \mathring{\bm \phi}'_{1;0}-\tfrac{1}{2}(\pb+4)(\pb+3) \breve{\mathbf{B}}_{\pb+2;2(\pb+1)} 
\cdot \mathring{\bm \phi}'_{2;0} \nonumber\\
&& \hspace{2cm} + \mathbf{S}_{2(\pb+1);\pb+1} \label{Order4c}
\end{eqnarray}
\end{subequations}
where
\begin{subequations}
\begin{eqnarray}
&& \mathbf{R}_{2(\pb+1);\pb+1}=(\pb+4)(\pb+3)\mathbf{C}_{4,2(\pb+1);\pb+1}\cdot \bigg[
\tfrac{1}{6}(\pb+2)\mathring{\mathbf{F}}_{3;4}\cdot
\breve{\phi}_{\pb+1;2(\pb+1)}  + \tfrac{1}{2}
\breve{\mathbf{F}}_{\pb+2;2(\pb+1)}\cdot \mathring{\bm\phi}_{2;4} \nonumber\\
&& \hspace{2cm} - \tfrac{1}{6}(\pb+2) \left(  \mathring{\mathbf{A}}_{3;4} \cdot
\breve{\bm \phi}'_{\pb+1;2(\pb+1)}
+\mathring{\mathbf{A}}^+_{3;4} \cdot
\breve{\bm \phi}_{\pb+1;2(\pb+1)} 
+\mathring{\mathbf{A}}^-_{3;4} \cdot
\breve{\bm \phi}_{\pb+1;2(\pb+1)}\right) \nonumber \\
&& \hspace{2cm} -\tfrac{1}{6} (\pb+2)\left(
\breve{\mathbf{A}}_{\pb+2;2(\pb+1)}\cdot \mathring{\bm \phi}'_{2;4} +
\breve{\mathbf{A}}^+_{\pb+2;2(\pb+1)}\cdot \mathring{\bm \phi}_{2;4}+
\breve{\mathbf{A}}^-_{\pb+2;2(\pb+1)}\cdot \mathring{\bm
  \phi}_{2;4}\right) \bigg], \label{FormulaR} 
\end{eqnarray} 
\begin{eqnarray}
&& \mathbf{S}_{2(\pb+1);\pb+1}=(\pb+4)(\pb+3)\mathbf{C}_{4,2(\pb+1);\pb+1} \cdot\bigg[
\tfrac{1}{6}(\pb+2)\mathring{\mathbf{G}}_{3;4}\cdot
\breve{\phi}_{\pb+1;2(\pb+1)}  + \tfrac{1}{2}
\breve{\mathbf{G}}_{\pb+2;2(\pb+1)}\cdot \mathring{\bm\phi}_{2;4} \nonumber \\
&& \hspace{2cm} - \tfrac{1}{6}(\pb+2) \left(  \mathring{\mathbf{B}}_{3;4} \cdot
\breve{\bm \phi}'_{\pb+1;2(\pb+1)}
+\mathring{\mathbf{B}}^+_{3;4} \cdot
\breve{\bm \phi}_{\pb+1;2(\pb+1)} 
+\mathring{\mathbf{B}}^-_{3;4} \cdot
\breve{\bm \phi}_{\pb+1;2(\pb+1)}\right) \nonumber \\
&& \hspace{2cm} -\tfrac{1}{6} (\pb+2)\left(
\breve{\mathbf{B}}_{\pb+2;2(\pb+1)}\cdot \mathring{\bm \phi}'_{2;4} +
\breve{\mathbf{B}}^+_{\pb+2;2(\pb+1)}\cdot \mathring{\bm \phi}_{2;4}+
\breve{\mathbf{B}}^-_{\pb+2;2(\pb+1)}\cdot \mathring{\bm
  \phi}_{2;4}\right) \bigg].
\label{FormulaS} 
\end{eqnarray} 
\end{subequations}

The initial data for the transport equations
\eqref{Order4b}-\eqref{Order4c} can be written as
\begin{equation}
\label{Data4}
\breve{\bm\phi}_{\pb+4,2(\pb+1)}(0)= \tilde{\bm\phi}_{\pb+4,2(\pb+1)}(0)+ \hat{\bm\phi}_{\pb+4,2(\pb+1)}(0),
\end{equation}
where $\tilde{\bm\phi}_{\pb+4,2(\pb+1)}$ depends solely on $m$ and
$\breve{w}_{\pb+1;2(\pb+1),k}$ (deviation from Schwarzschild) while
$\hat{\bm\phi}_{\pb+4,2(\pb+1)}$ contains contributions from the
 multipolar structure of the reference static data.  

\medskip
In order to discuss equations \eqref{Order4a}-\eqref{Order4c} we
introduce a further Ansatz. We write
\begin{subequations}
\begin{eqnarray}
&& \breve{\bm \upsilon}_{\pb+4,2(\pb+1)} = \tilde{\bm \upsilon}_{\pb+4,2(\pb+1)} + \hat{\bm \upsilon}_{\pb+4,2(\pb+1)}, \label{FurtherAnsatz1}\\ 
&& \breve{\bm \phi}_{\pb+4,2(\pb+1)} = \tilde{\bm \phi}_{\pb+4,2(\pb+1)} + \hat{\bm \phi}_{\pb+4,2(\pb+1)}, \label{FurtherAnsatz2}
\end{eqnarray}
\end{subequations}
where $\tilde{\bm \upsilon}_{\pb+4,2(\pb+1)}$ and $\tilde{\bm
\upsilon}_{\pb+4,2(\pb+1)}$ are the solutions to equations
\eqref{Order4a}-\eqref{Order4c} with
\begin{equation}
\mathbf{R}_{2(\pb+1);\pb+1}=0, \quad \mathbf{S}_{2(\pb+1);\pb+1}=0,
\quad \hat{\bm\phi}_{\pb+4;2(\pb+1)}(0)=0. \label{OldCase}
\end{equation}
The crucial observation is the following:

\begin{observation}
Equations \eqref{Order4a}-\eqref{Order4c} and the data \eqref{Data4}
satisfying \eqref{OldCase} are formally identical to the order $\pb+4$
transport equations analysed in \cite{Val10} for data which are 
Schwarzschildean up to order $\pb$. These equations have no solution
which is smooth at the critical sets.
\end{observation}

From the analysis in \cite{Val10} it follows that:
\begin{eqnarray}
&&\partial_\tau^{11} \tilde{\bm\phi}_{\pb+4;2(\pb+1)}=
m^3\breve{w}_{\pb+1;2(\pb+1)}(1-\tau)^{\pb-15}(1+\tau)^{\pb-15}
\nonumber \\
&& \hspace{4cm}\times\bigg( \tilde{\bm \varphi}_{\pb+4}(\tau)+ \tilde{\bm
    \varphi}_{\pb+4}^+(\tau) \ln(1+\tau) + \tilde{\bm
    \varphi}_{\pb+4}^-(\tau) \ln(1-\tau)\bigg),  \label{phitilde}
\end{eqnarray}
where the entries of $\tilde{\bm \varphi}$, $\tilde{\bm \varphi}^\pm$ are polynomials of degree $33$.  
\medskip
Substitution of the Ansatz
\eqref{FurtherAnsatz1}-\eqref{FurtherAnsatz2} into equations
\eqref{Order4b}-\eqref{Order4c} renders:
\begin{subequations}
\begin{eqnarray}
&& (\mathbf{E}+\mathbf{A})\cdot\hat{\bm\phi}'_{\pb+4;2(\pb+1)} +
\mathbf{A}_{\pb+4;2(\pb+1)} \cdot \hat{\bm \phi}_{\pb+4;2(\pb+1)} = \mathbf{R}_{2(\pb+1);\pb+1},\label{newOrder4b}\\
&& \mathbf{B}\cdot\hat{\bm\phi}'_{\pb+4;2(\pb+1)} +
\mathbf{B}_{\pb+4;2(\pb+1)} \cdot \hat{\bm \phi}_{\pb+4;2(\pb+1)} = \mathbf{S}_{2(\pb+1);\pb+1}, \label{newOrder4c}
\end{eqnarray}
\end{subequations}
with $\mathbf{R}_{2(\pb+1);\pb+1}$ and $\mathbf{S}_{2(\pb+1);\pb+1}$
as given in \eqref{FormulaR}-\eqref{FormulaS}. To conclude our
argument, one needs to analyse
the non-polynomial solutions to these equations. The detailed
structure of the entries in the non-homogeneous terms
$\mathbf{R}_{2(\pb+1);\pb+1}$ and $\mathbf{S}_{2(\pb+1);\pb+1}$ can be
obtained from the explicit calculations of \cite{FriKan00}. In
particular, one has that these terms are polynomials of order
$2\pb+10$ in $\tau$. More importantly, it can be explicitly verified
that
\begin{eqnarray*}
&& \partial_\tau^8 \,\mathbf{R}_{2(\pb+1);\pb+1} =
\tilde{\mathbf{R}}_{2(\pb+1);\pb+1} (\tau) (1-\tau)^{\pb-10} (1+\tau)^{\pb-10}, \\
&& \partial_\tau^8 \,\mathbf{S}_{2(\pb+1);\pb+1} = \tilde{\mathbf{S}}_{2(\pb+1);\pb+1} (\tau) (1-\tau)^{\pb-10} (1+\tau)^{\pb-10},
\end{eqnarray*}
where $\tilde{\mathbf{R}}_{2(\pb+1);\pb+1} (\tau)$ and $\tilde{\mathbf{S}}_{2(\pb+1);\pb+1} (\tau)$ have
entries which are polynomials of degree 18 in $\tau$. Thus, the
structure of the solutions to equations 
\eqref{newOrder4b}-\eqref{newOrder4c} is best analysed if one takes 8
$\tau$-derivatives of the equations. From here arguments similar to
those in \cite{Val10} allow to show that  
\begin{eqnarray}
&&\partial^8_\tau \hat{\phi}_{\pb+4;2(\pb+1)} =
\mathring{w}_{2;4}\;\breve{w}_{\pb+1;2(\pb+1)}(1-\tau)^{\pb-11}(1+\tau)^{\pb-11}
\nonumber \\
&& \hspace{4cm}\times\bigg( \hat{\bm \varphi}_{\pb+4}(\tau)+ \hat{\bm
    \varphi}_{\pb+4}^+(\tau) \ln(1+\tau) + \hat{\bm
    \varphi}_{\pb+4}^-(\tau) \ln(1-\tau)\bigg), \label{phihat}
\end{eqnarray}
with $\hat{\bm \varphi}_{\pb+4}(\tau)$, $\hat{\bm \varphi}^\pm_{\pb+4}(\tau)$  having entries which are
polynomials of degree 24 in $\tau$. 

\medskip
Integrating the expressions \eqref{phitilde} and \eqref{phihat} one
finds that the polynomials multiplying the $\ln(1\pm\tau)$ in
$\tilde{\phi}_{\pb+4;2(\pb+1)}$ are of degree $2\pb+14$ and have an
overall factor of $m^3 \breve{w}_{\pb+1;2(\pb+1)}$., whereas those 
in $\hat{\phi}_{\pb+4;2(\pb+1)}$ are of degree $2\pb+10$ and have an
overall factor $\mathring{w}_{2;4}\breve{w}_{\pb+1;2(\pb+1)}$. It
follows that the polyhomogeneous terms in
$\tilde{\phi}_{\pb+4;2(\pb+1)}$ and $\hat{\phi}_{\pb+4;2(\pb+1)}$
cannot cancel each other to produce a $\breve{\phi}_{\pb+4;2(\pb+1)}$
which is entirely polynomial.  Thus, one has the following:

\begin{proposition}
The solutions to the order $\pb+4$ transport equations
\eqref{Order4b}-\eqref{Order4c} with data given by equation \eqref{Data4} have no polynomial solutions unless
$\breve{w}_{\pb+1;2(\pb+1),k}=0$. If this condition is not satisfied
the solutions develop logarithmic singularities at $\tau=\pm 1$ and the
solutions are of class $C^\omega(-1,1)\cap C^{\pb+3}[-1,1]$. 
\end{proposition}

\section{The main result}
\label{Section:Conclusions}

The discussion in the previous sections is summarised in the following result.

\begin{proposition}
Given a time symmetric initial data set which in a neighbourhood
$\mathcal{B}_a$ of infinity is static up to order $p=\pb$ (in the
sense of Definition \ref{StaticUpToCertainOrder}), the solutions to the
transport equations for the orders $p=\pb+1,\; \pb+2,\;\pb+3$ are
polynomial in $\tau$, and hence, extend smoothly through the critical
sets $\mathcal{I}^\pm$. On the other hand, the solutions at order
$p=\pb+4$ contain logarithmic singularities which can be avoided if an
only if the initial data is, in fact, static up to order $p=\pb+1$.
\end{proposition}

As in the case of the analysis given in \cite{Val10}, one can use the
previous result to implement an induction argument in which the
explicit computer algebra calculations of \cite{Val04a} play the role
of the base step to obtain the main result on this article.

\begin{theorem}
The solution to the regular finite initial value problem at spatial
infinity for time symmetric initial data of the form given by
Definition  \ref{StaticUpToCertainOrder}  is
smooth through $\mathcal{I}^\pm$ if and only if the restriction of the
data to $\mathcal{I}^0$ coincides with the restriction of static data
at every order ---that is, if their jets of order $p$ coincide for all
$p$. Furthermore, the analyticity of the setting implies that the data
is exactly static in a neighbourhood of infinity.
\end{theorem}

As already mentioned in the introduction, in order to complete the
analysis of the solutions of the transport equations at the critical
sets for general analytic time symmetric initial data sets, one needs
(conformally invariant) conditions which reduce the initial data set
to data with a static massless part in the sense discussed in Section
\ref{Section:Data}. This problem will be analysed elsewhere.

\section*{Acknowledgements}
This research was funded by an EPSRC Advanced Research Fellowship. I
thank H. Friedrich, A. Ace{\~n}a, T. B\"ackdahl and C. L\"ubbe for
useful discussions.  

% Path in QM 
%\bibliography{/home/network/jav/tex/Newgrbib}
% Path in Ludovica
%\bibliography{/Users/Juan/Documents/tex/Newgrbib}

\end{document}